\documentclass[a4paper,UKenglish]{lipics-v2016}
\usepackage[noend]{algorithmic}
\usepackage{algorithm}
\usepackage{xspace}

\def\e#1{\emph{#1}}

\newcommand{\eat}[1]{}

\newcommand{\set}[1]{\{#1\}}
\def\eqdef{\stackrel{\textsf{\tiny def}}{=}}

\newcommand{\algname}[1]{{\sf #1}}
\def\myrulewidth{4.00in}

\def\therule{\rule{\myrulewidth}{0.2pt}}

\newenvironment{algseries}[2]
{\begin{figure}[#1]
\def\thecaption{\caption{#2}}
\vskip-0.8em\begin{tabular}{p{\myrulewidth}}\therule\end{tabular}\vskip0.2em}
{\thecaption\end{figure}}

\newenvironment{insidealg}[2]
{\normalsize\begin{insidecode}{#1}{#2}{Algorithm}}
{\end{insidecode}}

\newenvironment{insidesub}[2]
{\begin{insidecode}{#1}{#2}{Subroutine}}
{\end{insidecode}}

\newenvironment{insidecode}[3]
{
\begin{tabular}{p{\myrulewidth}}
\multicolumn{1}{c}{\rule{0mm}{3mm}{\bf #3} $\algname{#1}(\mbox{#2})$\vspace{-0.6em}}\\
\therule\vskip-0.8em\therule
\vspace{-1em}
\begin{algorithmic}[1]}
{\end{algorithmic}
\vskip-0.3em\therule
\end{tabular}}

\def\consts{\mathsf{Const}}

\newcommand{\depset}{\mathrm{\Delta}}

\def\phi{\varphi}
\def\signature{\mathcal{S}}

\newenvironment{repeatresult}[2]
{\vskip0.5em\par\textsc{#1} #2.\em}
{\vskip1em}

\def\partitle#1{\vskip0.3em\par\noindent\textbf{#1.}\,\,}

\def\val#1{\mathtt{#1}}

\newcommand{\mathsc}[1]{{\normalfont\textsc{#1}}}

\def\fdparam#1#2{\langle #1,#2\rangle}

\def\maxsrepfd#1#2{\mathsc{CRep}\fdparam{#1}{#2}}

\DeclareMathOperator*{\argmax}{argmax}

\begin{document}

\clubpenalty=10000 
\widowpenalty = 10000

\title{The Complexity of Computing a Cardinality Repair for Functional
  Dependencies}

\eat{
\numberofauthors{2} 
\author{
\alignauthor
Ester Livshits\\
      \affaddr{Technion - Israel Institute of Technology}\\
      \affaddr{Haifa 32000, Israel}
       \email{esterliv@cs.technion.ac.il}
\alignauthor
Benny Kimelfeld\\
      \affaddr{Technion - Israel Institute of Technology}\\
      \affaddr{Haifa 32000, Israel}
       \email{bennyk@cs.technion.ac.il}
}}

\author[1]{Ester Livshits}
\author[2]{Benny Kimelfeld}

\affil[1]{Technion - Israel Institute of Technology \\
Haifa 32000, Israel\\
\texttt{esterliv@cs.technion.ac.il}
}

\affil[2]{Technion - Israel Institute of Technology \\
Haifa 32000, Israel\\
\texttt{bennyk@cs.technion.ac.il}
}

\authorrunning{E.~Livshits and B.~Kimelfeld}

\maketitle
\begin{abstract}
  For a relation that violates a set of functional dependencies, we
  consider the task of finding a maximum number of pairwise-consistent
  tuples, or what is known as a ``cardinality repair.'' We present a
  polynomial-time algorithm that, for certain fixed relation schemas
  (with functional dependencies), computes a cardinality
  repair. Moreover, we prove that on any of the schemas not covered by
  the algorithm, finding a cardinality repair is, in fact, an NP-hard
  problem. In particular, we establish a dichotomy in the complexity
  of computing a cardinality repair, and we present an efficient
  algorithm to determine whether a given schema belongs to the
  positive side or the negative side of the dichotomy.
\end{abstract}

\keywords{Inconsistent Databases; Repairs; Cardinality Repairs; Functional Dependencies}

\def\athree{a_3}
\def\aiii{a_{\mathrm{iii}}}
\def\htwo{h_2}
\def\hthree{h_3}
\def\hii{h_{\mathrm{ii}}}
\def\mtwo{m_2}
\def\mi{m_\mathrm{i}}
\def\mone{m_1}
\def\nfour{n_4}

\section{Preliminaries}\label{sec:preliminaries}

\def\sabc{\signature_{\mathrm{rl}}}
\def\dabc{\depset_{\mathrm{rl}}}
\def\stk{\signature_{2\mathrm{fd}}}
\def\dtk{\depset_{2\mathrm{fd}}}
\def\stfd{\signature_{2\mathrm{r}}}
\def\dtfd{\depset_{2\mathrm{r}}}
\def\str{\signature_{\mathrm{tr}}}
\def\dtr{\depset_{\mathrm{tr}}}
\def\rtk{R_{2\mathrm{fd}}}
\def\rabc{R_{\mathrm{rl}}}
\def\rtfd{R_{2\mathrm{r}}}
\def\rtr{R_{\mathrm{tr}}}

We first present some basic terminology and notation that we use
throughout the paper.

\subsection{Relational Signatures and Instances}
We assume three infinite collections: \e{attributes} (column names),
\e{relation symbols} (table names), and \e{constants} (cell values).
A \e{heading} is a sequence $(A_1,\dots,A_k)$ of distinct attributes,
where $k$ is the \e{arity} of the heading. A \e{signature}
$\signature$ is a mapping from a finite set of relation symbols $R$ to
headings $\signature(R)$. We use the conventional notation
$R(A_1,\dots,A_k)$ to denote that $R$ is a relation symbol that is
assigned the heading $(A_1,\dots,A_k)$. An \e{instance} $I$ of a
signature $\signature$ maps every relation symbol $R(A_1,\dots,A_k)$
to a finite set, denoted $R^I$, of tuples $(c_1,\dots,c_k)$ where each
$c_i$ is a constant.  We may omit stating the signature $\signature$
of an instance $I$ when $\signature$ is clear from the context or
irrelevant.

Let $\signature$ be a signature, $R$ a relation symbol of
$\signature$, $I$ an instance of $\signature$, and $t$ be a tuple in
$R^I$. We refer to the expression $R(t)$ as a \e{fact of $I$}. By a
slight abuse of notation, we identify an instance $I$ with the set of
its facts. For example, $R(t)\in I$ denotes that $t$ is a tuple in
$R^I$.  As another example, $J\subseteq I$ means that
$R^J\subseteq R^I$ for every relation symbol $R$ of $\signature$; in
this case, we say that $J$ is \e{subinstance} of $I$.

\subsection{Functional Dependencies}

A \e{Functional Dependency} (\e{FD} for short) over a signature
$\signature$ is an expression of the form $R:X\rightarrow Y$, where
$R$ is a relation symbol and $X$ and $Y$ are sets of attributes of
$R$. When $R$ is clear from the context, we simply write
$X\rightarrow Y$.  We may also write $X$ and $Y$ by simply
concatenating the attribute symbols; for example, we may write
$AB\rightarrow C$ instead of $\set{A,B}\rightarrow\set{C}$ for the
relation symbol $R(A,B,C)$. An FD $X\rightarrow Y$ is \e{trivial} if
$Y\subseteq X$, and otherwise it is \e{nontrivial}. We say that an attribute $A$ in an FD $X\rightarrow Y$ is trivial if it holds that $A\in X$ and $A\in Y$. In this case, removing a trivial attribute from the FD means removing it from $Y$. For example, if we remove the trivial attributes from the FD $AB\rightarrow ACD$, the result is $AB\rightarrow CD$.

An instance $I$ \e{satisfies} an FD $R:X\rightarrow Y$ if for every
two facts $f$ and $g$ over $R$, if $f$ and $g$ agree on (i.e., have
the same constants in the position of) the attributes of $X$, then
they also agree on the attributes of $Y$. We say that $I$ satisfies a
set $\depset$ of FDs if $I$ satisfies every FD in $\Delta$; otherwise,
we say that $I$ \e{violates} $\Delta$. Two sets of FDs over the same
signature are \e{equivalent} if every instance that satisfies one also
satisfies the other. For example,
$\Delta=\set{R:A\rightarrow BC,R:C\rightarrow A}$ and
$\set{R:A\rightarrow C,R:C\rightarrow AB}$ are equivalent.  An FD
$X\rightarrow Y$ is \e{entailed} by $\depset$ (denoted by
$\depset\models X\rightarrow Y$) if for every instance $I$ over the
schema, if $I$ satisfies $\depset$, then it also satisfies
$X\rightarrow Y$.  We denote by $\depset_{|R}$ the restriction of
$\depset$ to the FDs over $R$ (i.e., those of the form
$R:X\rightarrow Y$).

Let $\depset$ be a set of FDs. We say that an FD, $(X_i\rightarrow Y_i)\in\depset$,
is a \e{local minimum} of $\depset$, if there is no other FD, $(X_j\rightarrow Y_j)\in\depset$, such that $X_j\subset X_i$. We say that the FD is a \e{global minimum} of $\depset$, if it holds that $X_i\subseteq X_j$ for every FD $(X_j\rightarrow Y_j)\in\depset$. 

An \e{FD schema} is a pair $(\signature,\depset)$, where $\signature$
is a signature and $\depset$ is a set of FDs over $\signature$.  Two
FD schemas $(\signature,\depset)$ and $(\signature',\depset')$ are
\e{equivalent} if $\signature=\signature'$ and $\depset$ is equivalent
to $\depset'$.  We say that an FD schema is a \e{chain} if for every two FDs $X_1\rightarrow Y_1$ and $X_2 \rightarrow Y_2$ over the same relation symbol, either $X_1 \subseteq X_2$ or $X_2 \subseteq X_1$~\cite{DBLP:conf/pods/LivshitsK17}.

\begin{table}
\def\arraystretch{1.4}
\caption{Specific FD schemas\label{table:special-schemas}}
\vskip0.5em
\centering
\begin{tabular}{c|c|c}
\textbf{FD Schema} & \textbf{Signature} & \textbf{FDs}\\\hline
$(\stk,\dtk)$ & $R(A,B,C)$ & $AB\rightarrow C$, $C\rightarrow B$\\
$(\sabc,\dabc)$ & $R(A,B,C)$ & $A\rightarrow B$, $B\rightarrow C$\\
$(\stfd,\dtfd)$ & $R(A,B,C)$ & $A\rightarrow C$, $B\rightarrow C$\\
$(\str,\dtr)$ & $R(A,B,C)$ & $AB\rightarrow C$, $AC\rightarrow B$, $BC\rightarrow A$\\
\hline
\end{tabular}
\end{table}

For example, Table~\ref{table:special-schemas} depicts
specific schemas that we refer to throughout the paper.
None of these schemas is a chain, while
the schema $\depset=\set{\emptyset\rightarrow A,B\rightarrow C}$ is a chain since it holds that $\emptyset\subset B$.

\subsection{Repairs}
Let $(\signature,\depset)$ be an FD schema and let $I$ be an inconsistent instance of $\signature$. We say that $J$ is a \e{subset repair} of $I$, or \e{s-repair} for short, if $J$ is a maximal consistent subinstance of $I$ (that is, $J$ does not violate any FD in $\depset$, and it is not possible to add another fact from $I\setminus J$ to $J$ without violating consistency) \cite{DBLP:conf/icdt/AfratiK09, TIPSTER98}. We say that $J$ is a \e{cardinality repair} of $I$, or \e{C-repair} for short, if $J$ is a maximum s-repair of $I$ (that is, there is no other subset repair of $I$ that contains more facts that $J$ does) \cite{DBLP:conf/icdt/LopatenkoB07}.

\section{Main Result}
In this section, we present our main result, which is a dichotomy for the problem of finding a C-repair of an inconsistent database. Note that since we only consider FD schemas, conflicting facts always belong to the same relation. Thus, if the schema contains two or more relations, we can solve the problem for each relation separately. Hence, our analysis can be restricted to single-relation schemas.

Let $(\signature,\depset)$ be an FD schema, and let $R(A_1,\dots,A_k)$ be the single relation in the schema. Let $A=\set{A_{i_1},\dots,A_{i_n}}$ be a subset of $\set{A_1,\dots,A_k}$. We denote by $\pi_{\overline{A}}(\signature)$ the projection of $\signature$ onto the attributes in $\set{A_1,\dots,A_k}\setminus A$. We also denote by $\pi_{\overline{A}}(\depset)$ the result of removing the attributes in $A$ from all the FDs in $\depset$. In addition, we denote by $\pi_{\overline{A}}(X\rightarrow Y)$ the result of removing the attributes in $A$ from the FD $(X\rightarrow Y)\in\depset$ (that is, $\pi_{\overline{A}}(X\rightarrow Y)=(X\setminus A)\rightarrow (Y\setminus A)$).

Let $(\signature,\depset)$ be an FD schema. We start by defining the following simplification steps:
\partitle{Simplification 1}
If some attribute $A_i$ appears on the left-hand side of all the FDs in $\depset$, remove the attribute $A_i$ from $\signature$ and from all the FDs in $\depset$. We denote the result by $(\pi_{\overline{\set{A_i}}}(\signature),\pi_{\overline{\set{A_i}}}(\depset))$.
\partitle{Simplification 2}
If $\depset$ contains an FD of the form $\emptyset\rightarrow X$, remove the attributes in $X$ from $\signature$ and from all the FDs in $\depset$. We denote the result by $(\pi_{\overline{X}}(\signature),\pi_{\overline{X}}(\depset))$.
\partitle{Simplification 3}
If $\depset$ contains two FDs $X_1\rightarrow Y_1$ and $X_2\rightarrow Y_2$, such that $X_1\subseteq Y_2$ and $X_2\subseteq Y_1$, and for each FD $Z\rightarrow W$ in $\depset$ it holds that $X_1\subseteq Z$ or $X_2\subseteq Z$, remove the attributes in $X_1\cup X_2$ from $\signature$ and from all the FDs in $\depset$. We denote the result by $(\pi_{\overline{X_1\cup X_2}}(\signature),\pi_{\overline{X_1\cup X_2}}(\depset))$.

\begin{example}
Let $(\signature,\depset)$ be an FD schema, such that $\depset=\set{\emptyset\rightarrow A, DB\rightarrow ACE, DC\rightarrow B, DB\rightarrow F}$. We can apply simplification $1$ to the schema, since it contains the FD $\emptyset\rightarrow A$. The result will be $\pi_{\overline{\set{A}}}(\depset)=\set{DB\rightarrow CE, DC\rightarrow B, DB\rightarrow F}$. Next, we can apply simplification $2$, since the attribute $D$ appears on the left-hand side of all the FDs in $\pi_{\overline{\set{A}}}(\depset)$. The result will be $\pi_{\overline{\set{A}}}(\pi_{\overline{\set{D}}}(\depset))=\set{B\rightarrow CE, C\rightarrow B, B\rightarrow F}$. Finally, since $\depset$ contains two FDs $B\rightarrow CE$ and $C\rightarrow B$, and it holds that $C\subset CE$, we can apply simplification $3$. The result will be The result will be $\pi_{\overline{\set{A}}}(\pi_{\overline{\set{D}}}(\pi_{\overline{\set{B,C}}}(\depset)))=\set{\emptyset\rightarrow E, \emptyset\rightarrow \emptyset, \emptyset\rightarrow F}$.
\end{example}

\begin{algseries}{t}{\label{alg:maxind} Can the problem of finding a C-repair be solved in polynomial time?}
\begin{insidealg}{CRepPTime}{$\signature,\depset$}
\WHILE{some simplification can be applied to $(\signature,\depset)$}
\STATE remove trivial FDs and attributes from $\depset$
\STATE apply the simplification to $(\signature,\depset)$
\ENDWHILE

\IF {$\depset=\emptyset$}
\STATE \textbf{return} true
\ELSE
\STATE \textbf{return} false
\ENDIF
\end{insidealg}
\end{algseries}

\def\mainthm{
Let $(\signature,\depset)$ be an FD schema. Then, the problem $\maxsrepfd{\signature}{\depset}$ can be solved in polynomial time if and only if $\algname{CRepPTime}(\signature,\depset)$ returns \e{true}.
}

\begin{theorem}\label{thm:main-maxsrep}
\mainthm
\end{theorem}

The algorithm $\algname{CRepPTime}$, depicted in Figure~\ref{alg:maxind}, starts with the given FD schema $(\signature,\depset)$, and at each step it tries to apply one of the simplifications to $(\signature,\depset)$. If at some point no simplification can be applied to $(\signature,\depset)$, there are two possible cases:
\begin{itemize}
\item $\depset$ is empty. In this case, there is a polynomial time algorithm for solving $\maxsrepfd{\signature}{\depset}$.
\item $\depset$ is not empty. In this case, $\maxsrepfd{\signature}{\depset}$ is NP-hard.
\end{itemize}

In the next sections we prove Theorem~\ref{thm:main-maxsrep}. 

\section{Finding a Cardinality Repair}
In this section we introduce a recursive algorithm for finding a
C-repair for a given instance $I$ over an FD schema $(\signature,
\depset)$. The algorithm is depicted in
Figures~\ref{alg:findcrep-main} and~\ref{alg:findcrep-sub}. If the
problem $\maxsrepfd{\signature}{\depset}$ can be solved in polynomial
time, the algorithm will return a C-repair, otherwise it will return
$\emptyset$. The algorithm's structure is similar to that of
$\algname{CRepPTime}$, and it uses the three subroutines:
$\algname{FindCRepS1}$, $\algname{FindCRepS2}$ and
$\algname{FindCRepS3}$. We will now introduce these three subroutines.

The subroutine $\algname{FindCRepS1}$ is used if simplification $1$ can be applied to $(\signature,\depset)$. It divides the given instance $I$ into blocks of facts that agree on the value of attribute $A_i$, and then finds a C-repair for each block separately, using the algorithm $\algname{FindCRep}$. Then, it returns the union of all those C-repairs. 

The subroutine $\algname{FindCRepS2}$ is used if simplification $2$ can be applied to $(\signature,\depset)$. It divides the given instance $I$ into blocks of facts that agree on the values of all the attributes in $X$, and then finds a C-repair for each block separately, using the algorithm $\algname{FindCRep}$. Then, the algorithm selects the C-repair that contains the most facts among those C-repairs and returns it. 

The subroutine $\algname{FindCRepS3}$ is used if simplification $3$ can be applied to $(\signature,\depset)$. It divides the given instance $I$ into blocks of facts that agree on the values of all the attributes in $X_1\cup X_2$, and then finds the C-repair for each block separately, using the algorithm $\algname{FindCRep}$. Then, the algorithm uses an existing polynomial time algorithm for finding the maximum weight matching in a bipartite graph $G_{X_1||X_2}$~\cite{DBLP:journals/jacm/EdmondsK72}. This graph has a node on its left-hand side for each possible set of values $x$ such that $f[X_1]=x$ for some fact $f\in I$. Similarly, it has a node on its right-hand side for each possible set of values $y$ such that $f[X_2]=y$ for some fact $f\in I$. The weight of each edge $(x,y)$ is the number of facts that appear in a C-repair of the block $B_{x||y}$ (the block that contains all the facts $f$ such that $f[X_1]=x$ and $f[X_2]=y$). The algorithm returns the subinstance that correspond to this maximum weighted matching (that is, the subinstance that contains the C-repair of each block $B_{xy}$ such that the edge $(x,y)$ belongs to the maximum matching).

As long as there exists a simplifaction that can be applied to $(\signature,\depset)$, the algorithm $\algname{FindCRep}$ applies this simplification to the schema and calls the corresponding subroutine on the result. If not simplification can be applied to $(\signature,\depset)$, then $\algname{FindCRep}$ returns the instance $I$ itseld if $\depset=\emptyset$, or $\emptyset$ otherwise. In the following sections we will prove the correctness of the algorithm $\algname{FindCRep}$ and Finally we will prove Theorem~\ref{thm:main-maxsrep}.

{
\begin{algseries}{t}{\label{alg:findcrep-main} Finding a cardinality repair (main)}
\begin{insidealg}{FindCRep}{$\signature,\depset,I$}
\STATE remove trivial FDs and attributes from $\depset$
\IF{$\depset=\emptyset$}
\STATE \textbf{return} $I$
\ENDIF
\IF{simplification $1$ can be applied to $(\signature,\depset)$}
\STATE \textbf{return} $\algname{FindCRepS1}(\signature,\depset,I)$
\ENDIF
\IF{simplification $2$ can be applied to $(\signature,\depset)$}
\STATE \textbf{return} $\algname{FindCRepS2}(\signature,\depset,I)$
\ENDIF
\IF{simplification $3$ can be applied to $(\signature,\depset)$}
\STATE \textbf{return} $\algname{FindCRepS3}(\signature,\depset,I)$
\ENDIF
\STATE \textbf{return} $\emptyset$
\end{insidealg}
\end{algseries}

\begin{algseries}{t}{\label{alg:findcrep-sub} Finding a cardinality repair (subroutines)}
\begin{insidesub}{FindCRepS1}{$\signature,\depset,I$}
\STATE $V:=\set{v\mid\mbox{$f[A_i]=v$ for some $f\in I$}}$
\FORALL{$v\in V$}
\STATE $B_v:=\set{f\in I\mid f[A_i]=v}$
\STATE $M_v:=\algname{FindCRep}(\pi_{\overline{\set{A_i}}}(\signature),\pi_{\overline{\set{A_i}}}(\depset),B_v)$
\ENDFOR
\STATE \textbf{return} $\bigcup\limits_{v\in V} M_v$
\end{insidesub}
\begin{insidesub}{FindCRepS2}{$\signature,\depset,I$}
\STATE $V:=\set{x\mid\mbox{$f[X]=x$ for some $f\in I$}}$
\FORALL{$x\in V$}
\STATE $B_x:=\set{f\in I\mid f[X]=x}$
\STATE $M_x:=\algname{FindCRep}(\pi_{\overline{X}}(\signature),\pi_{\overline{X}}(\depset),B_x)$
\ENDFOR
\STATE \textbf{return} $\argmax\limits_{x\in V} |M_x|$
\end{insidesub}
\begin{insidesub}{FindCRepS3}{$\signature,\depset,I$}
\STATE $V:=\set{x||y\mid\mbox{$f[X_1]=x$ and $f[X_2]=y$ for some $f\in I$}}$
\FORALL{$x||y\in V$}
\STATE $B_{xy}:=\set{f\in I\mid f[X_1]=x,f[X_2]=y}$
\STATE $M_{xy}:=\algname{FindCRep}(\pi_{\overline{X_1\cup X_2}}(\signature),\pi_{\overline{X_1\cup X_2}}(\depset),B_{xy})$
\STATE $w_{xy}:=|M_{xy}|$
\ENDFOR
\STATE $W:=\algname{FindMaxWeightMatch}(G_{X_1||X_2})$
\STATE $J:=\emptyset$
\FORALL{matches $\set{x,y}\in W$}
\STATE $J:=J\cup M_{xy}$
\ENDFOR
\STATE \textbf{return} $J$
\end{insidesub}
\end{algseries}
}

\section{Tractability Side}~\label{sec:tractability}

In this section, we prove, for each one of the three simplifications, that if the problem of finding a C-repair can be solved in polynomial time, using the algorithm $\algname{FindCRep}$, after applying the simplification to a schema $(\signature,\depset)$, then it can also be solved in polynomial time for the original schema $(\signature,\depset)$. More formally, we prove the following three lemmas.

\begin{lemma}\label{lemma:s1-ptime}
Let $(\signature,\depset)$ be an FD schema, such that simplification $1$ can be applied to $(\signature,\depset)$. Let $I$ be an instance of $\signature$. If $\maxsrepfd{\pi_{\overline{\set{A_i}}}(\signature)}{\pi_{\overline{\set{A_i}}}(\depset)}$ can be solved in polynomial time using $\algname{FindCRep}$, the problem $\maxsrepfd{\signature}{\depset}$ can be solved in polynomial time using $\algname{FindCRep}$ as well.
\end{lemma}

\begin{proof}
Assume that $\maxsrepfd{\pi_{\overline{\set{A_i}}}(\signature)}{\pi_{\overline{\set{A_i}}}(\depset)}$ can be solved in polynomial time using $\algname{FindCRep}$. That is, for each $J$, the algorithm $\algname{FindCRep}(\pi_{\overline{\set{A_i}}}(\signature),\pi_{\overline{\set{A_i}}}(\depset),J)$ returns a C-repair of $J$. We contend that $\maxsrepfd{\signature}{\depset}$ can also be solved in polynomial time using $\algname{FindCRep}$. Since the condition of line~4 of $\algname{FindCRep}$ is satisfied, the algorithm will call subroutine $\algname{FindCRepS1}$ and return the result. Thus, we have to prove that $\algname{FindCRepS1}(\signature,\depset,I)$ returns a C-repair of $I$.

Let $J$ be the result of $\algname{FindCRepS1}(\signature,\depset,I)$. We will start by proving that $J$ is consistent. Let us assume, by way of contradiction, that $J$ is not consistent. Thus, there are two facts $f_1$ and $f_2$ in $J$ that violate an FD $Z\rightarrow W$ in $\depset$. Since $A_i\in Z$, $f_1$ and $f_2$ agree on the value of attribute $A_i$, thus they belong to the same block $B_v$. By definition, there is an FD $(Z\setminus\set{A_i)}\rightarrow (W\setminus\set{A_i})$ in $\pi_{\overline{\set{A_i}}}(\depset)$. Clearly, the facts $f_1$ and $f_2$ agree on all the attributes in $Z\setminus\set{A_i}$, and since they also agree on the attribute $A_i$, there exists an attribute $B\in (W\setminus\set{A_i})$ such that $f_1[B]\neq f_2[B]$. Thus, $f_1$ and $f_2$ violate an FD in $\pi_{\overline{\set{A_i}}}(\depset)$, which is a contradiction to the fact that $\algname{FindCRep}(\pi_{\overline{\set{A_i}}}(\signature),\pi_{\overline{\set{A_i}}}(\depset),B_v)$ returns a C-repair of $B_v$ that contains both $f_1$ and $f_2$.

Next, we will prove that $J$ is a C-repair of $I$. Let us assume, by way of contradiction, that this is not the case. That is, there is another subset repair $J'$ of $I$, such that $J'$ contains more facts than $J$. In this case, there exists at least one value $v$ of attribute $A_i$, such that $J'$ contains more facts $f$ for which it holds that $f[A_i]=v$ than $J$. Let $\set{f_1,\dots,f_n}$ be the set of facts from $I$ for which it holds that $f[A_i]=v$ that appear in $J$, and let $\set{f_1,\dots,f_m}$ be the set of facts from $I$ for which it holds that $f[A_i]=v$ that appear in $J'$. It holds that $m>n$. We claim that $\set{f_1,\dots,f_m}$ is a subset repair of $B_v$, which is a contradiction to the fact that $\set{f_1,\dots,f_n}$ is a C-repair (that is, a C-repair) of $B_v$. Let us assume, by way of contradiction, that $\set{f_1,\dots,f_m}$ is not a subset repair of $B_v$. Thus, there exist two facts $f_{j_1}$ and $f_{j_2}$ in $\set{f_1,\dots,f_m}$ that violate an FD, $Z\rightarrow W$, in $\pi_{\overline{\set{A_i}}}(\depset)$. By definition, there is an FD $(Z\cup \set{A_i})\rightarrow (W\cup\set{A_i})$ in $\depset$, and since $f_{j_1}$ and $f_{j_2}$ agree on the value of attribute $A_i$, they clearly violate this FD, which is a contradiction to the fact that they both appear in $J'$ (which is a subset repair of $I$).

Clearly, if the the algorithm $\algname{FindCRep}$ solves the problem $\maxsrepfd{\pi_{\overline{\set{A_i}}}(\signature)}{\pi_{\overline{\set{A_i}}}(\depset)}$ in polynomial time, then it also solves the problem $\maxsrepfd{\signature}{\depset}$ in polynomial time, and that concludes our proof of the lemma.
\end{proof}

\begin{lemma}\label{lemma:s2-ptime}
Let $(\signature,\depset)$ be an FD schema, such that simplification $2$ can be applied to $(\signature,\depset)$. Let $I$ be an instance of $\signature$. If $\maxsrepfd{\pi_{\overline{X}}(\signature)}{\pi_{\overline{X}}(\depset)}$ can be solved in polynomial time using $\algname{FindCRep}$, the problem $\maxsrepfd{\signature}{\depset}$ can be solved in polynomial time using $\algname{FindCRep}$ as well.
\end{lemma}

\begin{proof}
Assume that $\maxsrepfd{\pi_{\overline{X}}(\signature)}{\pi_{\overline{X}}(\depset)}$ can be solved in polynomial time using $\algname{FindCRep}$. That is, for each $J$, the algorithm $\algname{FindCRep}(\pi_{\overline{X}}(\signature),\pi_{\overline{X}}(\depset),J)$ returns a C-repair of $J$. We contend that $\maxsrepfd{\signature}{\depset}$ can also be solved in polynomial time using $\algname{FindCRep}$. Note that the condition of line~4 cannot be satisfied, since there is no attribute that appears on the left-hand side of $\emptyset\rightarrow X$. Since the condition of line~6 of $\algname{FindCRep}$ is satisfied, the algorithm will call subroutine $\algname{FindCRepS2}$ and return the result. Thus, we have to prove that $\algname{FindCRepS2}(\signature,\depset,I)$ returns a C-repair of $I$.

Let $J$ be the result of $\algname{FindCRepS2}(\signature,\depset,I)$. We will start by proving that $J$ is not consistent. Thus, there are two facts $f_1$ and $f_2$ in $J$ that violate an FD $Z\rightarrow W$ in $\depset$. That is, $f_1$ and $f_2$ agree on all the attributes in $Z$, but do not agree on at least one attribute $B\in W$. Note that $f_1$ and $f_2$ agree on all the attributes in $X$ (since $\algname{FindCRepS2}(\signature,\depset,I)$ always returns a set of facts that belong to a single block). Thus, it holds that $B\not\in X$. By definition, there is an FD $(Z\setminus X)\rightarrow (W\setminus X)$ in $\pi_{\overline{X}}(\depset)$. Clearly, the facts $f_1$ and $f_2$ agree on all the attributes in $Z\setminus X$, but do not agree on the attribute $B\in (W\setminus X)$. Thus, $f_1$ and $f_2$ violate an FD in $\pi_{\overline{X}}(\depset)$, which is a contradiction to the fact that $\algname{FindCRep}(\pi_{\overline{X}}(\signature),\pi_{\overline{X}}(\depset),B_x)$ returns a C-repair of $B_x$ that contains both $f_1$ and $f_2$.

Next, we will prove that $J$ is a C-repair of $I$. Let us assume, by way of contradiction, that this is not the case. That is, there is another subset repair $J'$ of $I$, such that $J'$ contains more facts than $J$. Clearly, each subset repair of $I$ only contains facts that belong to a single block $B_x$ (since the FD $\emptyset\rightarrow X$ implies that all the facts must agree on the values of all the attributes in $X$). The instance $J$ is a C-repair of some block $B_x$. If $J'\subseteq B_x$, then we get a contradiction to the fact that $J$ is a C-repair of $B_x$. Thus, $J'$ contains facts from another block $B_{x'}$. In this case, the C-repair of $B_{x'}$ contains more facts than the C-repair of $B_x$, which is a contradiction to the fact that no block has a C-repair that contains more facts than $B_x$ does.

Clearly, if the the algorithm $\algname{FindCRep}$ solves the problem $\maxsrepfd{\pi_{\overline{X}}(\signature)}{\pi_{\overline{X}}(\depset)}$ in polynomial time, then it also solves the problem $\maxsrepfd{\signature}{\depset}$ in polynomial time, and that concludes our proof of the lemma.
\end{proof}

\begin{lemma}\label{lemma:s3-ptime}
Let $(\signature,\depset)$ be an FD schema, such that simplification $3$ can be applied to $(\signature,\depset)$. Let $I$ be an instance of $\signature$. If $\maxsrepfd{\pi_{\overline{X_1\cup X_2}}(\signature)}{\pi_{\overline{X_1\cup X_2}}(\depset)}$ can be solved in polynomial time using $\algname{FindCRep}$, then $\maxsrepfd{\signature}{\depset}$ can be solved in polynomial time using $\algname{FindCRep}$ as well.
\end{lemma}

\begin{proof}
Assume that $\maxsrepfd{\pi_{\overline{X_1\cup X_2}}(\signature)}{\pi_{\overline{X_1\cup X_2}}(\depset)}$ can be solved in polynomial time using $\algname{FindCRep}$. That is, $\algname{FindCRep}(\pi_{\overline{X_1\cup X_2}}(\signature),\pi_{\overline{X_1\cup X_2}}(\depset),J)$ returns a C-repair of $J$ for each $J$. We contend that $\maxsrepfd{\signature}{\depset}$ can also be solved in polynomial time using $\algname{FindCRep}$. Note that the condition of line~4 cannot be satisfied. Otherwise, there is an attribute $A_i$ that appears on the left-hand side of both $X_1\rightarrow Y_1$ and $X_2\rightarrow Y_2$. Since we always remove redundant attributes from the FDs in $\depset$ before calling $\algname{FindCRep}$, the attribute $A_i$ does not appear on the right-hand side of these FDs, and it does not hold that $X_1\subseteq Y_2$, which is a contradiction to the fact that simplification $3$ can be applied to $(\signature,\depset)$. The condition of line~6 cannot be satisfied as well, since neither $X_1\subseteq\emptyset$ nor $X_2\subseteq\emptyset$. The condition of line~8 on the other hand is satisfied, thus the algorithm will call subroutine $\algname{FindCRepS3}$ and return the result. Thus, we have to prove that $\algname{FindCRepS3}(\signature,\depset,I)$ returns a C-repair of $I$.

Let us denote by $J$ the result of $\algname{FindCRepS3}(\signature,\depset,I)$. We will start by proving that $J$ is consistent. Let $f_1$ and $f_2$ be two FDs in $I$. Note that it cannot be the case that $f_1[X_1]\neq f_2[X_1]$ but $f_1[X_2]= f_2[X_2]$ (or vice versa), since in this case the matching that we found for $G_{X_I||X_2}$ contains two edges $(x_1,y)$ and $(x_2,y)$, which is impossible. Moreover, if it holds that $f_1[X_1]\neq f_2[X_1]$ and $f_1[X_2]\neq f_2[X_2]$, then $f_1$ and $f_2$ do not agree on the left-hand side of any FD in $\depset$ (since we assumed that for each FD $Z\rightarrow W$ in $\depset$ it either holds that $X_1\subseteq Z$ or $X_2\subseteq Z$). Thus, $\set{f_1,f_2}$ satisfies all the FDs in $\depset$. Now, let us assume, by way of contradiction, that $J$ is not consistent. Thus, there are two facts $f_1$ and $f_2$ in $J$ that violate an FD $Z\rightarrow W$ in $\depset$. That is, $f_1$ and $f_2$ agree on all the attributes in $Z$, but do not agree on at least one attribute $B\in W$. As mentioned above, the only possible case is that $f_1[X_1]=f_2[X_1]=x$ and $f_1[X_2]=f_2[X_2]=y$. In this case, $f_1$ and $f_2$ belong to the same block $B_{xy}$, and they do not agree on an attribute $B\in (W\setminus(X_1\cup X_2))$. The FD $(Z\setminus(X_1\cup X_2)\rightarrow (W\setminus(X_1\cup X_2))$ belongs to $\pi_{\overline{X_1\cup X_2}}(\depset)$, and clearly $f_1$ and $f_2$ also vioalte this FD, which is a contradiction to the fact that $J$ only contains a C-repair of $B_{xy}$ and does not contain any other facts from this block.

Next, we will prove that $J$ is a C-repair of $I$. Let us assume, by way of contradiction, that this is not the case. That is, there is another subset repair $J'$ of $I$, such that $J'$ contains more facts than $J$. Note that the weight of the matching corresponding to $J$ is the total number of facts in $J$ (since the weight of each edge $(x,y)$ is the number of facts in the C-repair of the block $B_{x||y}$, and $J$ contains the C-repair of each block $B_{x||y}$, such that the edge $(x,y)$ belongs to the matching). Let $f_1$ and $f_2$ be two facts in $J'$. Note that it cannot be the case that that $f_1[X_1]= f_2[X_1]$ but $f_1[X_2]\neq f_2[X_2]$, since in this case, $\set{f_1,f_2}$ violates the FD $X_1\rightarrow Y_1$ (we recall that $X_2\subseteq Y_1$, thus the fact that $f_1[X_2]\neq f_2[X_2]$ implies that $f_1[Y_1]\neq f_2[Y_1]$). Hence, it either holds that $f_1[X_1]= f_2[X_1]$ and $f_1[X_2]= f_2[X_2]$ or $f_1[X_1]\neq f_2[X_1]$ and $f_1[X_2]\neq f_2[X_2]$. Therefore, $J'$ clearly corresponds to a matching of $G_{X_1||X_2}$ as well (the matching will contain an edge $(x,y)$ if there is a fact $f\in J'$, such that $f[X]=x$ and $f[Y]=y$).

Next, we claim that for each edge $(x,y)$ that belongs to the above matching, the subinstance $J'$ contains a C-repair of the block $B_{xy}$ w.r.t. $\pi_{\overline{X_1\cup X_2}}(\depset)$. Clearly, $J'$ cannot contain two facts $f_1$ and $f_2$ from $B_{xy}$ that violate an FD $Z\rightarrow W$ from $\pi_{\overline{X_1\cup X_2}}(\depset)$ (otherwise, $f_1$ and $f_2$ will also violate the FD $(Z\cup X_1\cup X_2)\rightarrow (W\cup X_1\cup X_2)$ from $\depset$, which is a contradiction to the fact that $J'$ is a subset repair of $I$). Thus, $J'$ contains a consistent set of facts from $B_{xy}$. If this set of facts is not a C-repair of $B_{xy}$, then we can replace this set of facts with a C-repair of $B_{xy}$. This will not break the consistency of $J'$ since these facts do not agree on the attributes in neither $X_1$ nor $X_2$ with any other fact in $J'$, and each FD $Z\rightarrow W$ in $\depset$ is such that $X_1\subseteq Z$ or $X_2\subseteq Z$. The result will be a repair of $I$ that contains more facts than $J'$, which is a contradiction to the fact that $J'$ is a C-repair of $I$. Therefore, for each edge $(x,y)$ that belongs to the above matching, $J'$ contains exactly $w_{xy}$ facts, which means that the weight of this matching is the total number of facts in $J'$. In this case, we found a matching of $G_{X_1||X_2}$ with a higher weight than the matching corresponding to $J$, which is a contradiction to the fact that $J$ corresponds to the maximum weighted matching of $G_{X_1||X_2}$. 

Clearly, if the the algorithm $\algname{FindCRep}$ solves the problem $\maxsrepfd{\pi_{\overline{X_1\cup X_2}}(\signature)}{\pi_{\overline{X_1\cup X_2}}(\depset)}$ in polynomial time, it also solves the problem $\maxsrepfd{\signature}{\depset}$ in polynomial time, and that concludes our proof of the lemma.
\end{proof}

\section{Hardness Side}
Our proof of hardness is based on the concept of a \e{fact-wise
  reduction}~\cite{DBLP:conf/pods/Kimelfeld12}.  Let
$(\signature,\depset)$ and $(\signature',\depset')$ be two FD
schemas. A \e{mapping} from $\signature$ to $\signature'$ is a
function $\mu$ that maps facts over $\signature$ to facts over
$\signature'$. We naturally extend a mapping $\mu$ to map instances
$I$ over $\signature$ to instances over $\signature'$ by defining
$\mu(I)$ to be $\set{\mu(f)\mid f\in I}$.  A \e{fact-wise reduction}
from $(\signature,\depset)$ to $(\signature',\depset')$ is a mapping
$\Pi$ from $\signature$ to $\signature'$ with the following
properties.
\begin{enumerate}
\item $\Pi$ is injective; that is, for all facts $f$ and $g$ over
  $\signature$, if $\Pi(f) = \Pi(g)$ then $f = g$.
\item $\Pi$ preserves consistency and inconsistency; that is, for
  every instance $I$ over $\signature$, the instance $\Pi(I)$
  satisfies $\depset'$ if and only if $I$ satisfies $\depset$.
\item $\Pi$ is computable in polynomial time.
\end{enumerate}

The following lemma is straightforward.
\begin{lemma}\label{lemma:fact-wise-sharpp}
  Let $(\signature,\depset)$ and $(\signature',\depset')$ be FD
  schemas, and suppose that there is a fact-wise reduction from
  $(\signature,\depset)$ to $(\signature',\depset')$. If the problem
  $\maxsrepfd{\signature}{\depset}$ is NP-hard, then so is
  $\maxsrepfd{\signature'}{\depset'}$.
\end{lemma}

 We first prove the hardness of $\maxsrepfd(\signature,\depset)$ for all the schemas that appear in Table~\ref{table:special-schemas}. Then, we prove the existence of fact-wise reductions from these schemas to other schemas. We will use all of these results in our proof of correctness for the algorithm $\algname{FindCRep}$.

\subsection{Hard Schemas}

We start by proving that $\maxsrepfd{\signature}{\depset}$ is NP-hard for four specific FD schemas.

\begin{lemma}\label{lemma:abc-cb-hard}
  The problem $\maxsrepfd{\stk}{\dtk}$ is NP-hard.
\end{lemma}

\begin{proof}
We construct a reduction from non-mixed CNF satisfiability to $\maxsrepfd{\stk}{\dtk}$.  The input to the first problem is a formula $\psi$ with the free variables 
 $x_1,\ldots,x_n$, such that  $\psi$
 has the form $c_1 \wedge  \cdots \wedge c_m$ where each $c_j$ is a clause. Each clause is a conjunction of variables from one of the following sets: \e{(a)} $\{x_i : i=1,\ldots,n\}$ or \e{(b)} $\{\neg x_i : i=1,\ldots,n\}$ (that is, each clause either contains only positive variables or only negative variables).
 The goal is to determine 
 if there exists an assignment
 $\tau: \{x_1,\ldots,x_n\} \rightarrow \{0,1\}$ that satisfies $\psi$. 
 Given such an input, 
 we will construct the input 
 $I$ for our problem as follows.
 For each $i=1,\ldots, n$ and $j=1,\ldots, m$, $I$ will contain the following facts:
 \begin{itemize}
 \item
 $\rtk(c_j,\val{1},x_i)$, if $c_j$ contains only positive variables and $x_i$ appears in $c_j$.
 \item
$\rtk(c_j,\val{0},x_i)$, if $c_j$ contains only negative variables and $\neg x_i$ appears in $c_j$.
 \end{itemize}
 We will now prove that there exists a satisfying assignment to $\psi$ if and only if the C-repair of $I$ contains exactly $m$ facts. 
 
\paragraph*{The ``if'' direction}
Assume that a C-repair $J$ of $I$ contains exactly $m$ facts. The FD $AB\rightarrow C$ implies that no subset repair  of $I$ contains two facts $\rtk(c_j,b_j,x_{i_1})$ and $\rtk(c_j,b_j,x_{i_2})$ such that $x_{i_1}\neq x_{i_2}$. Thus, each subset repair contains at most one fact $\rtk(c_j,b_j,x_i)$ for each $c_j$. Since $J$ contains exactly $m$ facts, it contains precisely one fact $\rtk(c_j,b_j,x_i)$ for each $c_j$. We will now define an assignment $\tau$ as follows: $\tau(x_i)\eqdef b_j$ if there exists a fact $\rtk(c_j,b_j,x_i)$ in $J$ for some $c_j$. Note that the FD $C\rightarrow B$ implies that no subset repair contains two facts $\rtk(c_{j_1},\val{1},x_i)$ and $\rtk(c_{j_2},\val{0},x_i)$, thus the assignment is well defined. Finally, sa mentioned above, $J$ contains a fact $\rtk(c_j,b_j,x_i)$ for each $c_j$. If $x_i$ appears in $c_j$ without negation, it holds that $b_1=1$, thus $\tau(x_i)\eqdef 1$ and $c_j$ is satisfied. Similarly, if $x_i$ appears in $c_j$ with negation, it holds that $b_j=0$, thus $\tau(x_i)\eqdef 0$ and $c_j$ is satisfied. Thus, each clause $c_j$ is satisfied by $\tau$ and we conclude that $\tau$ is a satisfying assingment of $\psi$.

\paragraph*{The ``only if'' direction}
Assume that $\tau: \{x_1,\ldots,x_n\} \rightarrow \{0,1\}$ is an assignment that satisfies $\psi$.
 We claim that the C-repair of $I$ containts exactly $m$ facts. Since $\tau$ is a satisfying assignment, for each clause $c_j$ there exists a variable $x_i\in c_j$, such that $\tau(x_i)=1$ if $x_i$ appears in $c_j$ without negation or $\tau(x_i)=0$ if it appears in $c_j$ with negation. Let us build an instance $J$ as follows. For each $c_j$ we will choose exactly one variable $x_i$ that satisfies the above and add the fact $\rtk(c_j,b_j,x_i)$ (where $\tau(x_i)=b_j$) to $J$. Since there are $m$ clauses, $J$ will contain exactly $m$ facts, thus it is only left to prove that $J$ is a subset repair. Let us assume, by way of contradiction, that $J$ is not a subset repair. As mentioned above, each subset repair can contain at most one fact $\rtk(c_j,b_j,x_i)$ for each $c_j$, thus $J$ is maximal. Moreover, since $J$ contains one fact for each $c_j$, no two facts violate the FD $AB\rightarrow C$. Thus, $J$ contains two facts $\rtk(c_{j_1},\val{1},x_i)$ and $\rtk(c_{j_2},\val{0},x_i)$, but this is a contradiction to the fact that $\tau$ is an assignment (that is, it cannot be the case that $\tau(x_i)=1$ and $\tau(x_i)=0$ as well). To conclude, $J$ is a subset repair that contains exactly $m$ facts, and since no subset repair can contain more than $m$ facts, $J$ is a C-repair.
\end{proof}

\begin{lemma}\label{lemma:ab-bc-hard}
  The problem $\maxsrepfd{\sabc}{\dabc}$ is NP-hard.
\end{lemma}

\begin{proof}
We construct a reduction from CNF satisfiability to $\maxsrepfd{\sabc}{\dabc}$.  The input to the first problem is a formula $\psi$ with the free variables 
 $x_1,\ldots,x_n$, such that  $\psi$
 has the form $c_1 \wedge  \cdots \wedge c_m$ where each $c_j$ is a clause. Each clause is a conjunction of variables from the set 
 $\{x_i, \neg x_i : i=1,\ldots,n\}$.
 The goal is to determine 
 if there exists an assignment
 $\tau: \{x_1,\ldots,x_n\} \rightarrow \{0,1\}$ that satisfies $\psi$. 
 Given such an input, 
 we will construct the input 
 $I$ for our problem as follows.
 For each $i=1,\ldots, n$ and $j=1,\ldots, m$, $I$ will contain the following facts:
 \begin{itemize}
 \item
 $\rabc(c_j,x_i,\val{1})$, if $x_i$ appears in $c_j$ without negation.
 \item
$\rabc(c_j,x_i,\val{0})$, if $x_i$ appears in $c_j$ with negation.
 \end{itemize}
 We will now prove that there exists a satisfying assignment to $\psi$ if and only if the C-repair of $I$ contains exactly $m$ facts. 
 
\paragraph*{The ``if'' direction}
Assume that a C-repair $J$ of $I$ contains exactly $m$ facts. The FD $A\rightarrow B$ implies that a subset repair cannot contain two facts $\rabc(c_j,x_{i_1},b_1)$ and $\rabc(c_j,x_{i_2},b_2)$ such that $x_{i_1}\neq x_{i_2}$. Moreover, the FD $B\rightarrow C$ implies that it cannot contain two facts $\rabc(c_j,x_i,\val{1})$ and $\rabc(c_j,x_i,\val{0})$. Thus, each subset repair contains at most one fact $\rabc(c_j,x_i,b_i)$ for each $c_j$. Since $J$ contains exactly $m$ facts, it contains precisely one fact $\rabc(c_j,x_i,b_i)$ for each $c_j$. We will now define an assignment $\tau$ as follows: $\tau(x_i)\eqdef b_i$ if there exists a fact $\rabc(c_j,x_i,b_i)$ in $J$ for some $c_j$. Note that no subset repair contains two facts $\rabc(c_{j_1},x_i,\val{1})$ and $\rabc(c_{j_2},x_i,\val{0})$, thus the assignment is well defined. Finally, sa mentioned above, $J$ contains a fact $\rabc(c_j,x_i,b_i)$ for each $c_j$. If $x_i$ appears in $c_j$ without negation, it holds that $b_i=1$, thus $\tau(x_i)\eqdef 1$ and $c_j$ is satisfied. Similarly, if $x_i$ appears in $c_j$ with negation, it holds that $b_i=0$, thus $\tau(x_i)\eqdef 0$ and $c_j$ is satisfied. Thus, each clause $c_j$ is satisfied by $\tau$ and we conclude that $\tau$ is a satisfying assingment of $\psi$.

\paragraph*{The ``only if'' direction}
Assume that $\tau: \{x_1,\ldots,x_n\} \rightarrow \{0,1\}$ is an assignment that satisfies $\psi$.
 We claim that the C-repair of $I$ containts exactly $m$ facts. Since $\tau$ is a satisfying assignment, for each clause $c_j$ there exists a variable $x_i\in c_j$, such that $\tau(x_i)=1$ if $x_i$ appears in $c_j$ without negation or $\tau(x_i)=0$ if it appears in $c_j$ with negation. Let us build an instance $J$ as follows. For each $c_j$ we will choose exactly one variable $x_i$ that satisfies the above and add the fact $\rabc(c_j,x_i,b_i)$ (where $\tau(x_i)=b_i$) to $J$. Since there are $m$ clauses, $J$ will contain exactly $m$ facts, thus it is only left to prove that $J$ is a subset repair. Let us assume, by way of contradiction, that $J$ is not a subset repair. As mentioned above, each subset repair can contain at most one fact $\rabc(c_j,x_i,b_i)$ for each $c_j$, thus $J$ is maximal. Moreover, since $J$ contains one fact for each $c_j$, no two facts violate the FD $A\rightarrow B$. Thus, $J$ contains two facts $\rabc(c_{j_1},x_i,\val{1})$ and $\rabc(c_{j_2},x_i,0)$, but this is a contradiction to the fact that $\tau$ is an assignment (that is, it cannot be the case that $\tau(x_i)=1$ and $\tau(x_i)=0$ as well). To conclude, $J$ is a subset repair that contains exactly $m$ facts, and since no subset repair can contain more than $m$ facts, $J$ is a C-repair.
\end{proof}

\begin{lemma}\label{lemma:ac-bc-hard}
  The problem $\maxsrepfd{\stfd}{\dtfd}$ is NP-hard.
\end{lemma}

\begin{proof}
We construct a reduction from CNF satisfiability to $\maxsrepfd{\stk}{\dtk}$. Given an input $\psi$ to the satisfiability problem, 
 we will construct the input 
 $I$ for our problem as follows.
 For each $i=1,\ldots, n$ and $j=1,\ldots, m$, $I$ will contain the following facts:
 \begin{itemize}
 \item
 $\rtfd(c_j,x_i,\langle x_i,\val{1} \rangle)$, if $x_i$ appears in $c_j$ without negation.
 \item
$\rtfd(c_j,x_i,\langle x_i,\val{0} \rangle)$, if $x_i$ appears in $c_j$ with negation.
 \end{itemize}
 We will now prove that there exists a satisfying assignment to $\psi$ if and only if the C-repair of $I$ contains exactly $m$ facts. 
 
\paragraph*{The ``if'' direction}
Assume that a C-repair $J$ of $I$ contains exactly $m$ facts. The FD $A\rightarrow C$ implies that no subset repair contains two facts $\rtfd(c_j,x_{i_1},\langle x_{i_1},b_{i_1} \rangle)$ and $\rtfd(c_j,x_{i_2},\langle x_{i_2},b_{i_2} \rangle)$ such that $x_{i_1}\neq x_{i_2}$. Moreover, the FD $B\rightarrow C$ implies that no subset repair contains two facts $\rtfd(c_j,x_i,\langle x_i,1 \rangle)$ and $\rtfd(c_j,x_i,\langle x_i,0 \rangle)$. Thus, each subset repair contains at most one fact $\rtfd(c_j,x_i,\langle x_i,b_i \rangle)$ for each $c_j$. Since $J$ contains exactly $m$ facts, it contains precisely one fact $\rtfd(c_j,x_i,\langle x_i,b_i \rangle)$ for each $c_j$. We will now define an assignment $\tau$ as follows: $\tau(x_i)\eqdef b_i$ if there exists a fact $\rtfd(c_j,x_i,\langle x_i,b_i \rangle)$ in $J$ for some $c_j$. Note that no subset repair contains two facts $\rtfd(c_{j_1},x_i,\langle x_i,\val{1} \rangle)$ and $\rtfd(c_{j_2},x_i,\langle x_i,\val{0} \rangle)$, thus the assignment is well defined. Finally, sa mentioned above, $J$ contains a fact $\rtfd(c_j,x_i,\langle x_i,b_i \rangle)$ for each $c_j$. If $x_i$ appears in $c_j$ without negation, it holds that $b_i=1$, thus $\tau(x_i)\eqdef 1$ and $c_j$ is satisfied. Similarly, if $x_i$ appears in $c_j$ with negation, it holds that $b_i=0$, thus $\tau(x_i)\eqdef 0$ and $c_j$ is satisfied. Thus, each clause $c_j$ is satisfied by $\tau$ and we conclude that $\tau$ is a satisfying assingment of $\psi$.

\paragraph*{The ``only if'' direction}
Assume that $\tau: \{x_1,\ldots,x_n\} \rightarrow \{0,1\}$ is an assignment that satisfies $\psi$.
 We claim that the C-repair of $I$ containts exactly $m$ facts. Since $\tau$ is a satisfying assignment, for each clause $c_j$ there exists a variable $x_i\in c_j$, such that $\tau(x_i)=1$ if $x_i$ appears in $c_j$ without negation or $\tau(x_i)=0$ if it appears in $c_j$ with negation. Let us build an instance $J$ as follows. For each $c_j$ we will choose exactly one variable $x_i$ that satisfies the above and add the fact $\rtfd(c_j,x_i,\langle x_i,b_i \rangle)$ (where $\tau(x_i)=b_i$) to $J$. Since there are $m$ clauses, $J$ will contain exactly $m$ facts, thus it is only left to prove that $J$ is a subset repair. Let us assume, by way of contradiction, that $J$ is not a subset repair. As mentioned above, each subset repair can contain at most one fact $\rtfd(c_j,x_i,\langle x_i,b_i \rangle)$ for each $c_j$, thus $J$ is maximal. Moreover, since $J$ contains one fact for each $c_j$, no two facts violate the FD $A\rightarrow X$. Thus, $J$ contains two facts $\rtfd(c_{j_1},x_i,\langle x_i,\val{1} \rangle)$ and $\rtfd(c_{j_2},x_i,\langle x_i,\val{0} \rangle)$, but this is a contradiction to the fact that $\tau$ is an assignment (that is, it cannot be the case that $\tau(x_i)=1$ and $\tau(x_i)=0$ as well). To conclude, $J$ is a subset repair that contains exactly $m$ facts, and since no subset repair contains more than $m$ facts, $J$ is a C-repair.
\end{proof}

\begin{lemma}\label{lemma:abc-acb-bca-hard}
  The problem $\maxsrepfd{\str}{\dtr}$ is NP-hard.
\end{lemma}

\begin{proof}
We construct a reduction from the problem of finding the maximum number of edge-disjoint triangles in a tripartite graph, which is known to be an NP-hard problem~\cite{DBLP:journals/dam/Colbourn84, DBLP:conf/stacs/HajirasoulihaJKS07}. The input to this problem is a tripartite graph $g$.
 The goal is to determine what is the maximum number of edge-disjoint triangles in $g$ (that is, no two triangles share an edge).
 We assume that $g$ contains three sets of nodes: $\set{a_1,\dots,a_n}$, $\set{b_1,\dots,b_l}$ and $\set{c_1,\dots,c_r}$. Given such an input, 
 we will construct the input 
 $I$ for our problem as follows.
 For each traingle in $g$ that consists of the nodes $a_i$, $b_j$, and $c_k$, $I$ will contain a fact $\rtr(a_i,b_j,c_k)$.
 We will now prove that the maximum number of edge-disjoint triangles in $g$ is $m$ if and only if the number of facts in a C-repair of $I$ is $m$. Thus, if we could solve the problem $\maxsrepfd{\str}{\dtr}$ in polynomial time, we could also solve the first problem in polynomial time. 
 
\paragraph*{The ``if'' direction}
Assume that a C-repair $J$ of $I$ contains exactly $m$ facts. The FD $AB\rightarrow C$ implies that a subset repair cannot contain two facts $\rtr(a_i,b_j,c_{k_1})$ and $\rtr(a_i,b_j,c_{k_2})$ such that $c_{k_1}\neq c_{k_2}$. Moreover, the FD $AC\rightarrow B$ implies that it cannot contain two facts $\rtr(a_i,b_{j_1},c_k)$ and $\rtr(a_i,b_{j_2},c_k)$ such that $b_{j_1}\neq b_{j_2}$, and the FD $BC\rightarrow A$ implies that it cannot contain two facts $\rtr(a_{i_1},b_j,c_k)$ and $\rtr(a_{i_2},b_j,c_k)$ such that $a_{i_1}\neq a_{i_2}$. Thus, the two triangles $(a_{i_1},b_{j_1},c_{k_1})$ and $(a_{i_2},b_{j_2},c_{k_2})$ in $g$ that correspond to two facts $\rtr(a_{i_1},b_{j_1},c_{k_1})$ and $\rtr(a_{i_1},b_{j_1},c_{k_1})$ in $J$, will not share an edge (they can only share a single node). Hence, there are at least $m$ edge-disjoint triangles in $g$. Let us assume, by way of contradiction, that $p$, the maximum number of edge-disjoint triangles in $g$, is greater than $m$. Let $\set{t_1,\dots,t_p}$ be a set of $p$ edge-disjoint triangles in $g$. In this case, we can build a subinstance $J'$ of $I$ as follows: for each triangle $(a_i,b_j,c_k)$  in $\set{t_1,\dots,t_p}$ we will add the fact $\rtr(a_i,b_j,c_k)$ to $J'$. Note that since the triangles are edge-disjoint, no two triangles share more than one node, thus no two facts in $J'$ agree on the value of more than one attribute. Therefore, $J'$ is consistent w.r.t. $\depset$, and we found a consistent subinstance of $I$ that contains more facts than $J$, which is a contradiction to the fact that $J$ is a C-repair of $I$. We can conclude the maximum number of edge-disjoint triangles in $g$ is $m$.

\paragraph*{The ``only if'' direction}
Assume that the maximum number of edge-disjoint triangles in $g$ is $m$. We can again build a consistent subinstance $J$ of $I$ as follows: for each triangle $(a_i,b_j,c_k)$  in $\set{t_1,\dots,t_p}$ we will add the fact $\rtr(a_i,b_j,c_k)$ to $J$. Thus, there is a consistent subinstance of $I$ that contains $m$ facts. Let us assume, by way of contradiction, that there is another consistent subinstance $J'$ of $I$ that contains more than $m$ facts. In this case, we can build a set of edge-disjoint traingles in $g$ as follows: for each fact $\rtr(a_i,b_j,c_k)$ in $J'$, we add the triangle $(a_i,b_j,c_k)$ to the set. Since no two facts in $J'$ agree on the value in more than one attribute, clearly, no two triangles in the set share an edge. Thus, we found a set of edge-disjoint triangles in $g$ that contains more than $m$ triangles, which is a contradiction to the fact that the maximum number of edge-disjoint triangles in $g$ is $m$. We can conclude that no consistent subinstance of $I$ contains more than $m$ facts, thus $J$ is a C-repair of $I$.
\end{proof}

\subsection{Fact-Wise Reductions}

Let $(\signature,\depset)$ be an FD schema. Note that as long as $\depset$ is a chain, we can always apply either simplification $1$ or simplification $2$ to the schema. Thus, if we reach a point where we cannot apply any simplifications to the schema, the set of FDs is not a chain. In this case, $\depset$ contains at least two local minima $X_1\rightarrow Y_1$ and $X_2\rightarrow Y_2$, and one of the following holds:
\begin{itemize}
\item $(X_1^+\setminus X_1)\cap X_2^+= \emptyset$ and $(X_2^+\setminus X_2)\cap X_1^+ = \emptyset$.
\item $(X_1^+\setminus X_1) \cap (X_2^+ \setminus X_2)\neq \emptyset$, $(X_1^+\setminus X_1)\cap X_2 = \emptyset$ and $(X_2^+ \setminus X_2)\cap X_1 = \emptyset$. 
\item $(X_1^+\setminus X_1)\cap X_2 \neq \emptyset$ and $(X_2^+ \setminus X_2)\cap X_1 = \emptyset$. 
\item $(X_1^+\setminus X_1)\cap X_2 \neq \emptyset$ and $(X_2^+ \setminus X_2)\cap X_1 \neq \emptyset$ and also $(X_1\setminus X_2)\subseteq (X_2^+\setminus X_2)$ and $(X_2\setminus X_1)\subseteq (X_1^+\setminus X_1)$. In this case, $\depset$ contains at least one more local minimum. Otherwise, for every FD $Z\rightarrow W$ in $\depset$ it holds that either $X_1\subseteq Z$ or $X_2\subseteq Z$. If $X_1\cap X_2\neq\emptyset$, then we can apply simplification $1$ to the schema, using an attribute from $X_1\cap X_2$. If $X_1\cap X_2=\emptyset$, then we can apply simplification $3$ to the schema. 
\item $(X_1^+\setminus X_1)\cap X_2 \neq \emptyset$ and $(X_2^+ \setminus X_2)\cap X_1 \neq \emptyset$ and also $(X_2\setminus X_1)\not\subseteq (X_1^+\setminus X_1)$. 
\end{itemize}
We will now prove that for each one of these cases there is a fact-wise reduction from one of the hard schemas.

\begin{lemma}\label{lemma:disjoint-reduction}
Let $(\signature,\depset)$ be an FD schema, such that $\depset$ contains two local minima $X_1\rightarrow Y_1$ and $X_2\rightarrow Y_2$, and it holds that $(X_1^+\setminus X_1)\cap X_2^+= \emptyset$ and $(X_2^+\setminus X_2)\cap X_1^+ = \emptyset$. Then, there is a fact-wise reduction from $\stfd$ to $\signature$.
\end{lemma}

\begin{proof}
First note that since $\depset$ is not a chain, there are indeed at least two local minima, $X_1\rightarrow Y_1$ and $X_2\rightarrow Y_2$, in $\depset$, such that $X_1\neq X_2$. We define a fact-wise reduction $\Pi:\stfd \rightarrow \signature$, using $X_1\rightarrow Y_1$ and $X_2\rightarrow Y_2$ and the constant $\odot \in \consts$.
Let $f = \rtfd(a,b,c)$ be a fact over $\stfd$ and let $\set{A_1,\dots,A_n}$ be the set of attributes in the single relation of $\signature$.
We define $\Pi $ as follows:
\[
\Pi (f)[A_k] \eqdef
\begin{cases}
\odot & \mbox{$A_k\in X_1\cap X_2$} \\ 
 a& \mbox{$A_k\in X_1 \setminus X_2$}\\ 
 b& \mbox{$A_k\in X_2 \setminus X_1$} \\
 \langle a,c\rangle& \mbox{$A_k\in X_1^+\setminus X_1$} \\
 \langle b,c\rangle& \mbox{$A_k\in X_2^+\setminus X_2$} \\
 \langle a,b\rangle& \mbox{otherwise}
\end{cases}
\]
It is left to show that $\Pi$ is a fact-wise reduction.
To do so, we prove that $\Pi$ is well defined, injective and preserves consistency and inconsistency.

\partitle{$\mathbf{\Pi}$ is well defined}
This is straightforward from the definition and the fact that $(X_1^+\setminus X_1)\cap X_2^+= \emptyset$ and $(X_2^+\setminus X_2)\cap X_1^+ = \emptyset$.

\partitle{$\mathbf{\Pi}$ is injective}
Let $f,f'$ be two facts, such that $f=\rtfd(a,b,c)$ and $f' = \rtfd(a',b',c')$.
Assume that $\Pi (f) = \Pi (f')$.  Let us denote
$\Pi (f)=R(x_1,\dots, x_n)$ and $\Pi (f')=R(x'_1,\dots, x'_n)$.
Note that $X_1 \setminus X_2$ and $X_2 \setminus X_1$ are not empty since $X_1\neq X_2$. Moreover, since both FDs are minimal, $X_1\not\subset X_2$ and $X_2\not\subset X_1$.
Therefore, there are $l$ and $p$ such that $\Pi (f)[A_l] = a$, $\Pi (f)[A_p]= b$. Furthermore, since $X_1\rightarrow Y_1$ and $X_2\rightarrow Y_2$ are not trivial, there are $m$ and $n$ such that $\Pi (f)[A_m]=\langle a,c\rangle$ and $\Pi (f)[A_n]=\langle b,c\rangle$.
Hence, $\Pi (f) = \Pi (f')$ implies that $\Pi (f)[A_l] = \Pi (f')[A_l]$, $\Pi (f)[A_p] = \Pi (f')[A_p]$, $\Pi (f)[A_m] = \Pi (f')[A_m]$ and also $\Pi (f)[A_n]= \Pi (f')[A_n]$. We obtain that
$a=a'$, $b=b'$ and $c=c'$, which implies $f=f'$.

\partitle{$\mathbf{\Pi}$ preserves consistency}
Let $f=\rtfd(a,b,c)$ and $f' = \rtfd(a',b',c')$ be two distinct facts.
We contend that  the set $\{f,f'\}$ is consistent w.r.t. $\dtfd$ if and only if the set $\{\Pi(f),\Pi(f')\}$ is consistent w.r.t. $\depset$.
\paragraph*{The ``if'' direction}
Assume that $\{f,f'\}$ is consistent w.r.t $\dtfd$. We prove that  $\{\Pi(f),\Pi(f')\}$ is consistent w.r.t $\depset$. First, note that each FD that contains an attribute $A_k\not\in (X_1^+\cup X_2^+)$ on its left-hand side is satisfied by $\{\Pi(f),\Pi(f')\}$, since $f$ and $f'$ cannot agree on both $A$ and $B$ (otherwise, the FD $A\rightarrow C$ implies that $f=f'$). Thus, from now on we will only consider FDs that do not contain an attribute $A_k\not\in (X_1^+\cup X_2^+)$ on their left-hand side. The FDs in $\dtfd$ imply that if $f$ and $f'$ agree on one of $\set{A,B}$ then they also agree on $C$, thus one of the following holds:
\begin{itemize}
\item $a\neq a'$, $b= b'$ and $c= c'$. In this case, $\Pi(f)$ and $\Pi(f')$ only agree on the attributes $A_k$ such that $A_k\in X_1\cap X_2$ or $A_k\in X_2\setminus X_1$ or $A_k\in X_2^+\setminus X_2$. That is, they only agree on the attributes $A_k$ such that $A_k\in X_2^+$. Thus, each FD that contains an attribute $A_k\not\in X_2^+$ on its left-hand side is satisfied. Moreover, any FD that contains only attributes $A_k\in X_2^+$ on its left-hand side, also contains only attributes $A_k\in X_2^+$ on its right-hand side (by definition of a closure), thus $\Pi(f)$ and $\Pi(f')$ agree on both the left-hand side and the right-hand side of such FDs and $\{\Pi(f),\Pi(f')\}$ satisfies all the FDs in $\depset$.
\item $a= a'$, $b\neq b'$ and $c= c'$. This case is symmetric to the previous one, thus a similar proof applies for this case as well.
\item $a\neq a'$, $b\neq b'$. In this case, $\Pi(f)$ and $\Pi(f')$ only agree on the attributes $A_k$ such that $A_k\in X_1\cap X_2$. Since $X_1\rightarrow Y_1$ and $X_2\rightarrow Y_2$ are minimal, there is no FD in $\depset$ that contains only attributes $A_k$ such that $A_k\in X_1\cap X_2$ on its left-hand side. Thus, $\Pi(f)$ and $\Pi(f')$ do not agree on the left-hand side of any FD in $\depset$ and $\{\Pi(f),\Pi(f')\}$ is consistent w.r.t. $\depset$.
\end{itemize}
This concludes our proof of the ``if'' direction.

\paragraph*{The ``only if'' direction}
Assume $\set{f,f'}$ is inconsistent w.r.t. $\dtfd$. We prove that $\{\Pi(f),\Pi(f')\}$ is inconsistent w.r.t. $\depset$.
Since $\set{f,f'}$ is inconsistent w.r.t. $\dtfd$ it either holds that $a=a'$ and $c\neq c'$ or $b=b'$ and $c\neq c'$. In the first case, $\Pi(f)$ and $\Pi(f')$ agree on the attributes on the left-hand side of the FD $X_1\rightarrow Y_1$, but do not agree on at least one attribute on its right-hand side (since the FD is not trivial). Similarly, in the second case, $\Pi(f)$ and $\Pi(f')$ agree on the attributes on the left-hand side of the FD $X_2\rightarrow Y_2$, but do not agree on at least one attribute on its right-hand side. Thus, $\{\Pi(f),\Pi(f')\}$ does not satisfy at least one of these FDs and $\{\Pi(f),\Pi(f')\}$ is inconsistent w.r.t. $\depset$.
\end{proof}

\begin{lemma}\label{lemma:disjoint-left-reduction}
Let $(\signature,\depset)$ be an FD schema, such that $\depset$ contains two local minima $X_1\rightarrow Y_1$ and $X_2\rightarrow Y_2$, and one of the following holds: 
\begin{itemize}
\item $(X_1^+\setminus X_1) \cap (X_2^+ \setminus X_2)\neq \emptyset$, $(X_1^+\setminus X_1)\cap X_2 = \emptyset$ and $(X_2^+ \setminus X_2)\cap X_1 = \emptyset$,
\item $(X_1^+\setminus X_1)\cap X_2 \neq \emptyset$ and $(X_2^+ \setminus X_2)\cap X_1 = \emptyset$.
\end{itemize}
Then, there is a fact-wise reduction from $\sabc$ to $\signature$.
\end{lemma}

\begin{proof}
We define a fact-wise reduction $\Pi:\sabc \rightarrow \signature$, using $X_1\rightarrow Y_1$ and $X_2\rightarrow Y_2$ and the constant $\odot \in \consts$.
Let $f = \sabc(a,b,c)$ be a fact over $\sabc$ and let $\set{A_1,\dots,A_n}$ be the set of attributes in the single relation of $\signature$.
We define $\Pi $ as follows:
\[
\Pi (f)[A_k] \eqdef
\begin{cases}
\odot & \mbox{$A_k\in X_1\cap X_2$} \\ 
 a& \mbox{$A_k\in X_1 \setminus X_2$}\\ 
 b& \mbox{$A_k\in X_2 \setminus X_1$} \\
 \langle a,c\rangle& \mbox{$A_k\in X_1^+\setminus X_1\setminus X_2^+$} \\
 \langle b,c\rangle& \mbox{$A_k\in X_2^+\setminus X_2$} \\
 a& \mbox{otherwise}
\end{cases}
\]
It is left to show that $\Pi$ is a fact-wise reduction.
To do so, we prove that $\Pi$ is well defined, injective and preserves consistency and inconsistency.

\partitle{$\mathbf{\Pi}$ is well defined}
This is straightforward from the definition and the fact that $(X_2^+ \setminus X_2)\cap X_1 = \emptyset$ in both cases.

\partitle{$\mathbf{\Pi}$ is injective}
Let $f,f'$ be two facts, such that $f=\rtfd(a,b,c)$ and $f' = \rabc(a',b',c')$.
Assume that $\Pi (f) = \Pi (f')$.  Let us denote
$\Pi (f)=R(x_1,\dots, x_n)$ and $\Pi (f')=R(x'_1,\dots, x'_n)$.
Note that $X_1 \setminus X_2$ and $X_2 \setminus X_1$ are not empty since $X_1\neq X_2$. Moreover, since both FDs are minimal, $X_1\not\subset X_2$ and $X_2\not\subset X_1$.
Therefore, there are $l$ and $p$ such that $\Pi (f)[A_l]= a$, $\Pi (f)[A_p] = b$. Furthermore, since $X_2\rightarrow Y_2$ is not trivial, there is at least one $m$ such that $\Pi (f)[A_m]=\langle b,c\rangle$.
Hence, $\Pi (f) = \Pi (f')$ implies that $\Pi (f)[A_l] = \Pi (f')[A_l]$, $\Pi (f)[A_p] = \Pi (f')[A_p]$ and $\Pi (f)[A_m] = \Pi (f')[A_m]$. We obtain that
$a=a'$, $b=b'$ and $c=c'$, which implies $f=f'$.

\partitle{$\mathbf{\Pi}$ preserves consistency}
Let $f=\rabc(a,b,c)$ and $f' = \rabc(a',b',c')$ be two distinct facts.
We contend that  the set $\{f,f'\}$ is consistent w.r.t. $\dabc$ if and only if the set $\{\Pi(f),\Pi(f')\}$ is consistent w.r.t. $\depset$.
\paragraph*{The ``if'' direction}
Assume that $\{f,f'\}$ is consistent w.r.t $\dabc$. We prove that  $\{\Pi(f),\Pi(f')\}$ is consistent w.r.t $\depset$. First, note that each FD that contains an attribute $A_k\not\in (X_1^+\cup X_2^+)$ on its left-hand side is satisfied by $\{\Pi(f),\Pi(f')\}$, since $f$ and $f'$ cannot agree on $A$ (otherwise, the FDs $A\rightarrow B$ and $B\rightarrow C$ imply that $f=f'$). Thus, from now on we will only consider FDs that do not contain an attribute $A_k\not\in (X_1^+\cup X_2^+)$ on their left-hand side. One of the following holds:
\begin{itemize}
\item $a\neq a'$, $b= b'$ and $c= c'$. In this case, $\Pi(f)$ and $\Pi(f')$ only agree on the attributes $A_k$ such that $A_k\in X_1\cap X_2$ or $A_k\in X_2\setminus X_1$ or $A_k\in X_2^+\setminus X_2$. That is, they only agree on the attributes $A_k$ such that $A_k\in X_2^+$. Thus, each FD that contains an attribute $A_k\not\in X_2^+$ on its left-hand side is satisfied. Moreover, any FD that contains only attributes $A_k\in X_2^+$ on its left-hand side, also contains only attributes $A_k\in X_2^+$ on its right-hand side (by definition of a closure), thus $\Pi(f)$ and $\Pi(f')$ agree on both the left-hand side and the right-hand side of such FDs and $\{\Pi(f),\Pi(f')\}$ satisfies all the FDs in $\depset$.
\item $a\neq a'$, $b\neq b'$. In this case, $\Pi(f)$ and $\Pi(f')$ only agree on the attributes $A_k$ such that $A_k\in X_1\cap X_2$. Since $X_1\rightarrow Y_1$ and $X_2\rightarrow Y_2$ are minimal, there is no FD in $\depset$ that contains only attributes $A_k$ such that $A_k\in X_1\cap X_2$ on its left-hand side. Thus, $\Pi(f)$ and $\Pi(f')$ do not agree on the left-hand side of any FD in $\depset$ and $\{\Pi(f),\Pi(f')\}$ is consistent w.r.t. $\depset$.
\end{itemize}
This concludes our proof of the ``if'' direction.

\paragraph*{The ``only if'' direction}
Assume $\set{f,f'}$ is inconsistent w.r.t. $\dabc$. We prove that $\{\Pi(f),\Pi(f')\}$ is inconsistent w.r.t. $\depset$.
Since $\set{f,f'}$ is inconsistent w.r.t. $\dabc$, one of the following holds:
\begin{itemize}
\item $a=a'$ and $b\neq b'$. For the first case of this lemma, since $(X_1^+\setminus X_1)\cap (X_2^+\setminus X_2) \neq \emptyset$ and since $(X_1^+\setminus X_1)\cap X_2 =(X_2^+\setminus X_2)\cap X_1 = \emptyset$, at least one attribute $A_k\in (X_1^+\setminus X_1)$ also belongs to $X_2^+\setminus X_2$ and it holds that $\Pi(f)[A_k]=\langle b,c\rangle$. For the second case of this lemma, since $(X_1^+\setminus X_1)\cap X_2 \neq \emptyset$, at least one attribute $A_k\in (X_1^+\setminus X_1)$ also belongs to $X_2$ and it holds that $\Pi(f)[A_k]=b$. Moreover, by definition of a closure, the FD $X_1\rightarrow A_k$ is implied by $\depset$. In both cases, the facts $\Pi(f)$ and $\Pi(f')$ agree on the attributes on the left-hand side of the FD $X_1\rightarrow A_k$, but do not agree on the right-hand side of this FD. If two facts do not satisfy an FD that is implied by a set $\depset$ of FDs, they also do not satisfy $\depset$, thus $\{\Pi(f),\Pi(f')\}$ is inconsistent w.r.t. $\depset$.
\item $a=a'$, $b=b'$ or $c\neq c'$. For the first case of this lemma, as mentioned above, there is an attribute $A_k\in (X_1^+\setminus X_1)$ such that $\Pi(f)[A_k]=\langle b,c\rangle$. Moreover, by definition of a closure, the FDs $X_1\rightarrow A_k$ is implied by $\depset$. The facts $\Pi(f)$ and $\Pi(f')$ agree on the attributes on the left-hand side of the FD $X_1\rightarrow A_k$, but do not agree on the right-hand side of this FD. For the second case of this lemma, since $(X_2^+\setminus X_2)\cap X_1 = \emptyset$ and since the FD $X_2\rightarrow Y_2$ is not trivial, there is at least one attribute $A_k\in (X_2^+\setminus X_2)$ such that $\Pi(f)[A_k]=\langle b,c \rangle$. Furthermore, the FDs $X_2\rightarrow A_k$ is implied by $\depset$. The facts $\Pi(f)$ and $\Pi(f')$ again agree on the attributes on the left-hand side of the FD $X_2\rightarrow A_k$, but do not agree on the right-hand side of this FD. If two facts do not satisfy an FD that is implied by a set $\depset$ of FDs, they also do not satisfy $\depset$, thus $\{\Pi(f),\Pi(f')\}$ is inconsistent w.r.t. $\depset$.
\item $a\neq a'$, $b=b'$ and $c\neq c'$. In this case, $\Pi(f)$ and $\Pi(f')$ agree on the attributes on the left-hand side of the FD $X_2\rightarrow Y_2$, but do not agree on the right-hand side of this FD (since the FD is not trivial and contains at least one attribute $A_k$ such that $\Pi(f)[A_k]=\langle b,c\rangle$ on its right-hand side). Therfore, $\{\Pi(f),\Pi(f')\}$ is inconsistent w.r.t. $\depset$.
\end{itemize}
This concludes our proof of the ``only if'' direction.
\end{proof}

\begin{lemma}\label{lemma:triangle-reduction}
Let $(\signature,\depset)$ be an FD schema, such that $\depset$ contains three local minima $X_1\rightarrow Y_1$, $X_2\rightarrow Y_2$ and $X_k\rightarrow Y_k$.
Then, there is a fact-wise reduction from $\str$ to $\signature$.
\end{lemma}

\begin{proof}
We define a fact-wise reduction $\Pi:\str \rightarrow \signature$, using $X_1\rightarrow Y_1$, $X_2\rightarrow Y_2$ and $X_k\rightarrow Y_k$ and the constant $\odot \in \consts$.
Let $f = \str(a,b,c)$ be a fact over $\str$ and let $\set{A_1,\dots,A_n}$ be the set of attributes in the single relation of $\signature$.
We define $\Pi $ as follows:
\[
\Pi (f)[A_k] \eqdef
\begin{cases}
 \odot& \mbox{$A_k\in X_1\cap X_2\cap X_k$}\\
 a& \mbox{$A_k\in (X_1\cap X_2)\setminus X_k$}\\
 b& \mbox{$A_k\in (X_1\cap X_k)\setminus X_2$}\\ 
 c& \mbox{$A_k\in (X_2\cap X_k)\setminus X_1$}\\ 
 \langle a,b\rangle& \mbox{$A_k\in X_1\setminus X_2\setminus X_k$} \\
 \langle a,c\rangle& \mbox{$A_k\in X_2\setminus X_1\setminus X_k$} \\
 \langle b,c\rangle& \mbox{$A_k\in X_k\setminus X_1\setminus X_2$} \\
 \langle a,b,c\rangle& \mbox{otherwise}
\end{cases}
\]
It is left to show that $\Pi$ is a fact-wise reduction.
To do so, we prove that $\Pi$ is well defined, injective and preserves consistency and inconsistency.

\partitle{$\mathbf{\Pi}$ is well defined}
This is straightforward from the definition.

\partitle{$\mathbf{\Pi}$ is injective}
Let $f,f'$ be two facts, such that $f=\rtfd(a,b,c)$ and $f' = \rabc(a',b',c')$.
Assume that $\Pi (f) = \Pi (f')$.  Let us denote
$\Pi (f)=R(x_1,\dots, x_n)$ and $\Pi (f')=R(x'_1,\dots, x'_n)$.
Note that $X_1$ contains at least one attribute that does not belong to $X_k$ (otherwise, it holds that $X_1\subseteq X_k$, which is a contradiction to the fact that $X_k$ is minimal). Thus, there exists an attribute $A_l$ such that either $\Pi (f)[A_l]=a$ or $\Pi (f)[A_l]=\langle a,b\rangle$. Similarly, $X_k$ contains at least one attribute that does not belong to $X_2$. Thus, there exists an attribute $A_p$ such that either $\Pi (f)[A_p]=b$ or $\Pi (f)[A_p]=\langle b,c\rangle$. Finally, $X_2$ contains at least one attribute that does not belong to $X_1$. Thus, there exists an attribute $A_r$ such that either $\Pi (f)[A_r]=c$ or $\Pi (f)[A_r]=\langle a,c\rangle$.
Hence, $\Pi (f) = \Pi (f')$ implies that $\Pi (f)[A_l] = \Pi (f')[A_l]$, $\Pi (f)[A_p] = \Pi (f')[A_p]$ and $\Pi (f)[A_r] = \Pi (f')[A_r]$. We obtain that
$a=a'$, $b=b'$ and $c=c'$, which implies $f=f'$.

\partitle{$\mathbf{\Pi}$ preserves consistency}
Let $f=\rabc(a,b,c)$ and $f' = \rabc(a',b',c')$ be two distinct facts.
We contend that  the set $\{f,f'\}$ is consistent w.r.t. $\dabc$ if and only if the set $\{\Pi(f),\Pi(f')\}$ is consistent w.r.t. $\depset$.
\paragraph*{The ``if'' direction}
Assume that $\{f,f'\}$ is consistent w.r.t $\dabc$. We prove that  $\{\Pi(f),\Pi(f')\}$ is consistent w.r.t $\depset$. Note that $f$ and $f'$ cannot agree on more than one attribute (otherwise, they will violate at least one FD in $\dabc$). Thus, $\Pi(f)$ and $\Pi(f')$ may only agree on attributes that appear in $X_1\cap X_2\cap X_k$ and in one of $X_1\cap X_2$, $X_1\cap X_k$ or $X_2\cap X_k$. As mentioned above, $X_1$ contains at least one attribute that does not belong to $X_k$, thus no FD in $\depset$ contains only attributes from $X_1\cap X_2\cap X_k$ and $X_1\cap X_k$ on its left-hand side (otherwise, $X_1$ will not be minimal). Similarly, no FD in $\depset$ contains only attributes from $X_1\cap X_2\cap X_k$ and $X_2\cap X_k$ on its left-hand side and no FD in $\depset$ contains only attributes from $X_1\cap X_2\cap X_k$ and $X_1\cap X_2$ on its left-hand side. Therefore $\Pi(f)$ and $\Pi(f')$ do not agree on the left-hand side of any FD in $\depset$, and $\{\Pi(f),\Pi(f')\}$ is consistent w.r.t. $\depset$.

\paragraph*{The ``only if'' direction}
Assume $\set{f,f'}$ is inconsistent w.r.t. $\dabc$. We prove that $\{\Pi(f),\Pi(f')\}$ is inconsistent w.r.t. $\depset$.
Since $\set{f,f'}$ is inconsistent w.r.t. $\dabc$, $f$ and $f'$ agree on two attributes, but do not agree on the third one. Thus, one of the following holds:
\begin{itemize}
\item $a=a'$, $b=b'$ and $c\neq c'$. In this case, $\Pi(f)$ and $\Pi(f')$ agree on all of the attributes that appear on the left-hand side of $X_1\rightarrow Y_1$. Since this FD is not trivial, it must contain on its right-hand side an attribute $A_k$ such that $A_k\not\in X_1$. That is, there is at least one attribute $A_k$ that appears on the right-hand side of $X_1\rightarrow Y_1$ such that one of the following holds: \e{(a)} $\Pi(f)[A_k]=c, $\e{(b)} $\Pi(f)[A_k]=\langle a,c\rangle$, \e{(c)} $\Pi(f)[A_k]=\langle b,c\rangle$ or \e{(d)} $\Pi(f)[A_k]=\langle a,b,c\rangle$. Hence, $\Pi(f)$ and $\Pi(f')$ do not satisfy the FD $X_1\rightarrow Y_1$ and $\{\Pi(f),\Pi(f')\}$ is inconsistent w.r.t. $\depset$.
\item $a=a'$, $b\neq b'$ and $c= c'$. This case is symmetric to the first one.  $\Pi(f)$ and $\Pi(f')$ agree on all of the attributes that appear on the left-hand side of $X_2\rightarrow Y_2$, but do not agree on at least one attribute that appears on the right-hand side of the FD.
\item $a\neq a'$, $b=b'$ and $c= c'$. This case is also symmetric to the first one. Here, $\Pi(f)$ and $\Pi(f')$ agree on the left-hand side, but not on the right-hand side of the FD $X_k\rightarrow Y_k$.
\end{itemize}
\end{proof}

\begin{lemma}\label{lemma:last-reduction}
Let $(\signature,\depset)$ be an FD schema, such that $\depset$ contains two local minima $X_1\rightarrow Y_1$ and $X_2\rightarrow Y_2$, and the following holds:
\begin{itemize}
\item $(X_1^+\setminus X_1)\cap X_2\neq\emptyset$ and $(X_2^+\setminus X_2)\cap X_1\neq\emptyset$,
\item $(X_2\setminus X_1)\not\subseteq (X_1^+\setminus X_1)$.
\end{itemize}
Then, there is a fact-wise reduction from $\stk$ to $\signature$.
\end{lemma}

\begin{proof}
We define a fact-wise reduction $\Pi:\stk \rightarrow \signature$, using $X_1\rightarrow Y_1$, $X_2\rightarrow Y_2$ and the constant $\odot \in \consts$.
Let $f = \str(a,b,c)$ be a fact over $\stk$ and let $\set{A_1,\dots,A_n}$ be the set of attributes in the single relation of $\signature$.
We define $\Pi $ as follows:
\[
\Pi (f)[A_k] \eqdef
\begin{cases}
\odot & \mbox{$A_k\in X_1\cap X_2$} \\ 
 c& \mbox{$A_k\in X_1\setminus X_2$}\\ 
 b& \mbox{$A_k\in (X_2\setminus X_1) \cap (X_1^+\setminus X_1)$}\\ 
 \langle a,b \rangle& \mbox{$A_k\in (X_2\setminus X_1) \setminus (X_1^+\setminus X_1)$}\\ 
 \langle b,c\rangle& \mbox{$A_k\in (X_1^+\setminus X_1)\setminus (X_2\setminus X_1)$}\\
 \langle a,b,c\rangle& \mbox{otherwise}
\end{cases}
\]
It is left to show that $\Pi$ is a fact-wise reduction.
To do so, we prove that $\Pi$ is well defined, injective and preserves consistency and inconsistency.

\partitle{$\mathbf{\Pi}$ is well defined}
This is straightforward from the definition.

\partitle{$\mathbf{\Pi}$ is injective}
Let $f,f'$ be two facts, such that $f=\rtk(a,b,c)$ and $f' = \rtk(a',b',c')$.
Assume that $\Pi (f) = \Pi (f')$.  Let us denote
$\Pi (f)=R(x_1,\dots, x_n)$ and $\Pi (f')=R(x'_1,\dots, x'_n)$.
Since the FD $X_2\rightarrow Y_2$ is a local minimum, it holds that $X_1\not\subseteq X_2$. Thus, there is an attribute that appears in $X_1$, but does not appear in $X_2$. Moreover, it holds that $(X_2\setminus X_1)\not\subseteq (X_1^+\setminus X_1)$, thus $X_2\setminus X_1$ contains at least one attribute that does not appear in $X_1^+\setminus X_1$.
Therefore, there are $l$ and $p$ such that $\Pi (f)[A_l] = c$, $\Pi (f)[A_p] = \langle a,b\rangle$. 
Hence, $\Pi (f) = \Pi (f')$ implies that $\Pi (f)[A_l] = \Pi (f')[A_l]$ and $\Pi (f)[A_p] = \Pi (f')[A_p]$. We obtain that
$a=a'$, $b=b'$ and $c=c'$, which implies $f=f'$.

\partitle{$\mathbf{\Pi}$ preserves consistency}
Let $f=\rtk(a,b,c)$ and $f' = \rtk(a',b',c')$ be two distinct facts.
We contend that  the set $\{f,f'\}$ is consistent w.r.t. $\dtk$ if and only if the set $\{\Pi(f),\Pi(f')\}$ is consistent w.r.t. $\depset$.
\paragraph*{The ``if'' direction}
Assume that $\{f,f'\}$ is consistent w.r.t $\dtk$. We prove that  $\{\Pi(f),\Pi(f')\}$ is consistent w.r.t $\depset$. One of the following holds:
\begin{itemize}
\item $b\neq b'$ and $c\neq c'$. In this case, $\Pi(f)$ and $\Pi(f')$ only agree on the attributes $A_k$ such that $A_k\in X_1\cap X_2$. Since $X_1\rightarrow Y_1$ and $X_2\rightarrow Y_2$ are local minima, there is no FD in $\depset$ that contains on its left-hand side only attributes $A_k$ such that $A_k\in X_1\cap X_2$. Thus, $\Pi(f)$ and $\Pi(f')$ do not agree on the left-hand side of any FD in $\depset$ and $\{\Pi(f),\Pi(f')\}$ is consistent w.r.t. $\depset$.
\item $a\neq a'$, $b= b'$ and $c=c'$. Note that in this case $\Pi(f)$ and $\Pi(f')$ agree on all of the attributes that belong to $X_1^+$, and only on these attributes. Any FD in $\depset$ that contains only attributes from $X_1^+$ on its left-hand side, also contains only attributes from $X_1^+$ on its right-hand side. Thus, $\Pi(f)$ and $\Pi(f')$ satisfy all the FDs in $\depset$.
\item $a\neq a'$, $b= b'$ and $c\neq c'$. In this case, $\Pi(f)$ and $\Pi(f')$ only agree on the attributes $A_k$ such that $A_k\in X_1\cap X_2$ or $A_k\in (X_2\setminus X_1) \cap (X_1^+\setminus X_1)$. Since the FD $X_2\rightarrow Y_2$ is a local minimum, and since $X_2$ contains an attributes that does not belong to $X_1^+\setminus X_1$, no FD in $\depset$ contains on its left-hand side only attributes $A_k$ such that $A_k\in X_1\cap X_2$ or $A_k\in (X_2\setminus X_1) \cap (X_1^+\setminus X_1)$. Thus,  $\Pi(f)$ and $\Pi(f')$ do not agree on the left-hand side of any FD in $\depset$ and $\{\Pi(f),\Pi(f')\}$ is consistent w.r.t. $\depset$.
\end{itemize}
This concludes our proof of the ``if'' direction.

\paragraph*{The ``only if'' direction}
Assume $\set{f,f'}$ is inconsistent w.r.t. $\dtk$. We prove that $\{\Pi(f),\Pi(f')\}$ is inconsistent w.r.t. $\depset$.
Since $\set{f,f'}$ is inconsistent w.r.t. $\dtk$, one of the following holds:
\begin{itemize}
\item $a=a$, $b= b'$ and $c\neq c'$. In this case, $\Pi(f)$ and $\Pi(f')$ agree on all of the attributes that appear in $X_2$. Since $X_2\rightarrow Y_2$ is not trivial, there is an attribute $A_k$ in $Y_2$ that does not belong to $X_2$. That is, one of the following holds: \e{(a)} $\Pi(f)[A_k]=c$, \e{(b)} $\Pi(f)[A_k]=\langle b,c \rangle$ or \e{(c)} $\Pi(f)[A_k]=\langle a,b,c \rangle$. Thus, $\{\Pi(f),\Pi(f')\}$ violates the FD $X_2\rightarrow Y_2$ and it is inconsistent w.r.t. $\depset$.
\item $b\neq b'$ and $c= c'$. In this case, $\Pi(f)$ and $\Pi(f')$ agree on all of the attributes that appear in $X_1$. Since $X_1\rightarrow Y_1$ is not trivial, there is an attribute $A_k$ in $Y_1$ that does not belong to $X_1$. That is, one of the following holds: \e{(a)} $\Pi(f)[A_k]=b$, \e{(b)} $\Pi(f)[A_k]=\langle a,b \rangle$, \e{(c)} $\Pi(f)[A_k]=\langle b,c \rangle$ or \e{(d)} $\Pi(f)[A_k]=\langle a,b,c \rangle$. Thus, $\{\Pi(f),\Pi(f')\}$ violates the FD $X_1\rightarrow Y_1$ and it is again inconsistent w.r.t. $\depset$.
\end{itemize}
This concludes our proof of the ``only if'' direction.
\end{proof}

Next, we prove, for each one of the three simplifications, that if the problem of finding the C-repair is hard after applying the simplification to a schema $(\signature,\depset)$, then it is also hard for the original schema $(\signature,\depset)$.

\begin{lemma}\label{lemma:s1-hard}
Let $(\signature,\depset)$ be an FD schema. If simplification $1$ can be applied to $(\signature,\depset)$ and $\pi_{\overline{\set{A_i}}}(\depset)\neq\emptyset$, then there is a fact-wise reduction from $(\pi_{\overline{\set{A_i}}}(\signature),\pi_{\overline{\set{A_i}}}(\depset))$ to $(\signature,\depset)$.
\end{lemma}

\begin{proof}
We define a fact-wise reduction $\Pi:\pi_{\overline{\set{A_i}}}(\signature) \rightarrow \signature$, using the constant $\odot \in \consts$.
Let $\set{A_1,\dots,A_m}$ be the set of attributes in the single relation of $\signature$. Let $f$ be a fact over $\pi_{\overline{\set{A_i}}}(\signature)$.
We define $\Pi $ as follows:
\[
\Pi (f)[A_k] \eqdef
\begin{cases}
\odot & \mbox{$A_k=A_i$} \\ 
f[A_k]& \mbox{otherwise}
\end{cases}
\]
It is left to show that $\Pi$ is a fact-wise reduction.
To do so, we prove that $\Pi$ is well defined, injective and preserves consistency and inconsistency.

\partitle{$\mathbf{\Pi}$ is well defined}
Every attribute in $\set{A_1,\dots,A_m}\setminus\set{A_i}$ also appears in $\pi_{\overline{\set{A_i}}}(\signature)$. Thus, for each attribute $A_k\in\set{A_1,\dots,A_m}\setminus\set{A_i}$ (that is, for each attribute $A_k\neq A_i$), $f[A_k]$ is a valid value. Therefore, $\Pi$ is well defined.

\partitle{$\mathbf{\Pi}$ is injective}
Let $f,f'$ be two distinct facts over $\pi_{\overline{\set{A_i}}}(\signature)$. Since $f\neq f'$, there exists an attribute $A_j$ in $\set{A_1,\dots,A_k}\setminus\set{A_i}$, such that $f[A_j]\neq f'[A_j]$. Thus, it also holds that $\Pi(f)[A_j]\neq \Pi(f')[A_j]$ and $\Pi(f)\neq \Pi(f')$.

\partitle{$\mathbf{\Pi}$ preserves consistency}
Let $f,f'$ be two distinct facts over $\pi_{\overline{\set{A_i}}}(\signature)$.
We contend that the set $\{f,f'\}$ is consistent w.r.t. $\pi_{\overline{\set{A_i}}}(\depset)$ if and only if the set $\{\Pi(f),\Pi(f')\}$ is consistent w.r.t. $\depset$.
\paragraph*{The ``if'' direction}
Assume that $\{f,f'\}$ is consistent w.r.t $\pi_{\overline{\set{A_i}}}(\depset)$. We prove that  $\{\Pi(f),\Pi(f')\}$ is consistent w.r.t $\depset$.  Let us assume, by way of contradiction, that $\{\Pi(f),\Pi(f')\}$ is inconsistent w.r.t $\depset$. That is, there exists an FD $X\rightarrow Y$ in $\depset$, such that $\Pi(f)$ and $\Pi(f')$ agree on all the attributes in $X$, but do not agree on at least one attribute $B$ in $Y$. Clearly, it holds that $B\neq A_i$ (since $\Pi(f)[A_i]=\Pi(f')[A_i]=\odot$). Note that the FD $X\setminus\set{A_i}\rightarrow Y\setminus\set{A_i}$ belongs to $\pi_{\overline{\set{A_i}}}(\depset)$. Since $\Pi(f)[A_k]=f[A_k]$ and $\Pi(f')[A_k]=f'[A_k]$ for each attribute $A_k\neq A_i$, the facts $f$ and $f'$ also agree on all the attributes on the left-hand side of the FD $X\setminus\set{A_i}\rightarrow Y\setminus\set{A_i}$, but do not agree on the attribute $B$ that belongs to $Y\setminus\set{A_i}$. Thus, $f$ and $f'$ violate an FD in $\depset$, which is a contradiction to the fact that the set $\{f,f'\}$ is consistent w.r.t. $\pi_{\overline{\set{A_i}}}(\depset)$.

\paragraph*{The ``only if'' direction}
Assume $\set{f,f'}$ is inconsistent w.r.t. $\pi_{\overline{\set{A_i}}}(\depset)$. We prove that $\{\Pi(f),\Pi(f')\}$ is inconsistent w.r.t. $\depset$.
Since $\set{f,f'}$ is inconsistent w.r.t. $\pi_{\overline{\set{A_i}}}(\depset)$, there exists an FD $X\rightarrow Y$ in $\pi_{\overline{\set{A_i}}}(\depset)$, such that $f$ and $f'$ agree on all the attributes on the left-hand side of the FD, but do not agree on at least one attribute $B$ on its right-hand side. The FD $X\cup\set{A_i}\rightarrow Y\cup\set{A_i}$ belongs to $\depset$, and since $\Pi(f)[A_k]=f[A_k]$ and $\Pi(f')[A_k]=f'[A_k]$ for each attribute $A_k\neq A_i$, and $\Pi(f)[A_i]=\Pi(f')[A_i]=\odot$, it holds that $\Pi(f)$ and $\Pi(f')$ agree on all the attributes in $X\cup\set{A_i}$, but do not agree on the attribute $B\in (Y\cup\set{A_i)}$, thus $\{\Pi(f),\Pi(f')\}$ is inconsistent w.r.t. $\depset$.
\end{proof}

\begin{lemma}\label{lemma:s2-hard}
Let $(\signature,\depset)$ be an FD schema. If simplification $2$ can be applied to $(\signature,\depset)$ and $\pi_{\overline{X}}(\depset)\neq\emptyset$, then there is a fact-wise reduction from $(\pi_{\overline{X}}(\signature),\pi_{\overline{X}}(\depset))$ to $(\signature,\depset)$.
\end{lemma}

\begin{proof}
We define a fact-wise reduction $\Pi:\pi_{\overline{X}}(\signature) \rightarrow \signature$, using the constant $\odot \in \consts$.
Let $\set{A_1,\dots,A_m}$ be the set of attributes in the single relation of $\signature$. Let $f$ be a fact over $\pi_{\overline{X}}(\signature)$.
We define $\Pi $ as follows:
\[
\Pi (f)[A_k] \eqdef
\begin{cases}
\odot & \mbox{$A_k\in X$} \\ 
f[A_k]& \mbox{otherwise}
\end{cases}
\]
It is left to show that $\Pi$ is a fact-wise reduction.
To do so, we prove that $\Pi$ is well defined, injective and preserves consistency and inconsistency.

\partitle{$\mathbf{\Pi}$ is well defined}
Every attribute in $\set{A_1,\dots,A_m}\setminus X$ also appears in $\pi_{\overline{X}}(\signature)$. Thus, for each attribute $A_k\in\set{A_1,\dots,A_m}\setminus X$ (that is, for each attribute $A_k\not\in X$), $f[A_k]$ is a valid value. Therefore, $\Pi$ is well defined.

\partitle{$\mathbf{\Pi}$ is injective}
Let $f,f'$ be two distinct facts over $\pi_{\overline{X}}(\signature)$. Since $f\neq f'$, there exists an attribute $A_j$ in $\set{A_1,\dots,A_k}\setminus X$, such that $f[A_j]\neq f'[A_j]$. Thus, it also holds that $\Pi(f)[A_j]\neq \Pi(f')[A_j]$ and $\Pi(f)\neq \Pi(f')$.

\partitle{$\mathbf{\Pi}$ preserves consistency}
Let $f,f'$ be two distinct facts over $\pi_{\overline{X}}(\signature)$.
We contend that the set $\{f,f'\}$ is consistent w.r.t. $\pi_{\overline{X}}(\depset)$ if and only if the set $\{\Pi(f),\Pi(f')\}$ is consistent w.r.t. $\depset$.
\paragraph*{The ``if'' direction}
Assume that $\{f,f'\}$ is consistent w.r.t $\pi_{\overline{X}}(\depset)$. We prove that  $\{\Pi(f),\Pi(f')\}$ is consistent w.r.t $\depset$.  Let us assume, by way of contradiction, that $\{\Pi(f),\Pi(f')\}$ is inconsistent w.r.t $\depset$. That is, there exists an FD $Z\rightarrow W$ in $\depset$, such that $\Pi(f)$ and $\Pi(f')$ agree on all the attributes in $Z$, but do not agree on at least one attribute $B$ in $W$. Clearly, it holds that $B\not\in X$ (since $\Pi(f)[A_i]=\Pi(f')[A_i]=\odot$ for each attribute $A_i\in X$). Note that the FD $(Z\setminus X)\rightarrow (W\setminus X)$ belongs to $\pi_{\overline{X}}(\depset)$. Since $\Pi(f)[A_k]=f[A_k]$ and $\Pi(f')[A_k]=f'[A_k]$ for each attribute $A_k\not\in X$, the facts $f$ and $f'$ also agree on all the attributes on the left-hand side of the FD $(Z\setminus X)\rightarrow (W\setminus X)$, but do not agree on the attribute $B$ that belongs to $W\setminus X$. Thus, $f$ and $f'$ violate an FD in $\depset$, which is a contradiction to the fact that the set $\{f,f'\}$ is consistent w.r.t. $\pi_{\overline{X}}(\depset)$.

\paragraph*{The ``only if'' direction}
Assume $\set{f,f'}$ is inconsistent w.r.t. $\pi_{\overline{X}}(\depset)$. We prove that $\{\Pi(f),\Pi(f')\}$ is inconsistent w.r.t. $\depset$.
Since $\set{f,f'}$ is inconsistent w.r.t. $\pi_{\overline{X}}(\depset)$, there exists an FD $Z\rightarrow W$ in $\pi_{\overline{X}}(\depset)$, such that $f$ and $f'$ agree on all the attributes on the left-hand side of the FD, but do not agree on at least one attribute $B$ on its right-hand side. The FD $(Z\cup X)\rightarrow (W\cup X)$ belongs to $\depset$, and since $\Pi(f)[A_k]=f[A_k]$ and $\Pi(f')[A_k]=f'[A_k]$ for each attribute $A_k\not\in X$, and $\Pi(f)[A_k]=\Pi(f')[A_k]=\odot$ for each attribute $A_k\in X$, it holds that $\Pi(f)$ and $\Pi(f')$ agree on all the attributes in $Z\cup X$, but do not agree on the attribute $B\in (W\cup X)$, thus $\{\Pi(f),\Pi(f')\}$ is inconsistent w.r.t. $\depset$.
\end{proof}

\begin{lemma}\label{lemma:s3-hard}
Let $(\signature,\depset)$ be an FD schema. If simplification $3$ can be applied to $(\signature,\depset)$ and $\pi_{\overline{X_1\cup X_2}}(\depset)$ is not empty, there is a fact-wise reduction from $(\pi_{\overline{X_1\cup X_2}}(\signature),\pi_{\overline{X_1\cup X_2}}(\depset))$ to $(\signature,\depset)$.
\end{lemma}

\begin{proof}
We define a fact-wise reduction $\Pi:\pi_{\overline{X_1\cup X_2}}(\signature) \rightarrow \signature$, using the constant $\odot \in \consts$.
Let $\set{A_1,\dots,A_m}$ be the set of attributes in the single relation of $\signature$. Let $f$ be a fact over $\pi_{\overline{X_1\cup X_2}}(\signature)$.
We define $\Pi $ as follows:
\[
\Pi (f)[A_k] \eqdef
\begin{cases}
\odot & \mbox{$A_k\in X_1\cup X_2$} \\ 
f[A_k]& \mbox{otherwise}
\end{cases}
\]
It is left to show that $\Pi$ is a fact-wise reduction.
To do so, we prove that $\Pi$ is well defined, injective and preserves consistency and inconsistency.

\partitle{$\mathbf{\Pi}$ is well defined}
Every attribute in $\set{A_1,\dots,A_m}\setminus (X_1\cup X_2)$ also appears in $\pi_{\overline{X_1\cup X_2}}(\signature)$. Thus, for each attribute $A_k\in\set{A_1,\dots,A_m}\setminus (X_1\cup X_2)$, $f[A_k]$ is a valid value. Therefore, $\Pi$ is well defined.

\partitle{$\mathbf{\Pi}$ is injective}
Let $f,f'$ be two distinct facts over the signature $\pi_{\overline{X_1\cup X_2}}(\signature)$. Since $f\neq f'$, there exists an attribute $A_j$ in $\set{A_1,\dots,A_k}\setminus (X_1\cup X_2)$, such that $f[A_j]\neq f'[A_j]$. Thus, it also holds that $\Pi(f)[A_j]\neq \Pi(f')[A_j]$ and $\Pi(f)\neq \Pi(f')$.

\partitle{$\mathbf{\Pi}$ preserves consistency}
Let $f,f'$ be two distinct facts over $\pi_{\overline{X_1\cup X_2}}(\signature)$.
We contend that the set $\{f,f'\}$ is consistent w.r.t. $\pi_{\overline{X_1\cup X_2}}(\depset)$ if and only if the set $\{\Pi(f),\Pi(f')\}$ is consistent w.r.t. $\depset$.
\paragraph*{The ``if'' direction}
Assume that $\{f,f'\}$ is consistent w.r.t $\pi_{\overline{X_1\cup X_2}}(\depset)$. We prove that  $\{\Pi(f),\Pi(f')\}$ is consistent w.r.t $\depset$.  Let us assume, by way of contradiction, that $\{\Pi(f),\Pi(f')\}$ is inconsistent w.r.t $\depset$. That is, there exists an FD $Z\rightarrow W$ in $\depset$, such that $\Pi(f)$ and $\Pi(f')$ agree on all the attributes in $Z$, but do not agree on at least one attribute $B$ in $W$. Clearly, it holds that $B\not\in (X_1\cup X_2)$ (since $\Pi(f)[A_i]=\Pi(f')[A_i]=\odot$ for each attribute $A_i\in (X_1\cup X_2)$). Note that the FD $(Z\setminus (X_1\cup X_2))\rightarrow (W\setminus (X_1\cup X_2))$ belongs to $\pi_{\overline{X_1\cup X_2}}(\depset)$. Since $\Pi(f)[A_k]=f[A_k]$ and $\Pi(f')[A_k]=f'[A_k]$ for each attribute $A_k\not\in (X_1\cup X_2)$, the facts $f$ and $f'$ also agree on all the attributes on the left-hand side of the FD $(Z\setminus (X_1\cup X_2))\rightarrow (W\setminus (X_1\cup X_2))$, but do not agree on the attribute $B$ that belongs to $W\setminus (X_1\cup X_2)$. Thus, $f$ and $f'$ violate an FD in $\depset$, which is a contradiction to the fact that the set $\{f,f'\}$ is consistent w.r.t. $\pi_{\overline{X_1\cup X_2}}(\depset)$.

\paragraph*{The ``only if'' direction}
Assume $\set{f,f'}$ is inconsistent w.r.t. $\pi_{\overline{X_1\cup X_2}}(\depset)$. We prove that $\{\Pi(f),\Pi(f')\}$ is inconsistent w.r.t. $\depset$.
Since $\set{f,f'}$ is inconsistent w.r.t. $\pi_{\overline{X_1\cup X_2}}(\depset)$, there exists an FD $Z\rightarrow W$ in the set $\pi_{\overline{X_1\cup X_2}}(\depset)$, such that $f$ and $f'$ agree on all the attributes on the left-hand side of the FD, but do not agree on at least one attribute $B$ on its right-hand side. The FD $(Z\cup X_1\cup X_2)\rightarrow (W\cup X_1\cup X_2)$ belongs to $\depset$, and since $\Pi(f)[A_k]=f[A_k]$ and $\Pi(f')[A_k]=f'[A_k]$ for each attribute $A_k\not\in (X_1\cup X_2)$, and $\Pi(f)[A_k]=\Pi(f')[A_k]=\odot$ for each attribute $A_k\in (X_1\cup X_2)$, it holds that $\Pi(f)$ and $\Pi(f')$ agree on all the attributes in $Z\cup X_1\cup X_2$, but do not agree on the attribute $B\in (W\cup X_1\cup X_2)$, thus $\{\Pi(f),\Pi(f')\}$ is inconsistent w.r.t. $\depset$.
\end{proof}

\section{Proof of Main Result}
In this section we prove Theorem~\ref{thm:main-maxsrep}. throughout this section we denote the set $X_i^+\setminus X_i$ by $X_i^*$. We start by proving the following lemma.

\begin{lemma}\label{lemma:ind-basis}
Let $(\signature,\depset)$ be an FD schema, such that no simplification can be applied to $(\signature,\depset)$. Then, the problem $\maxsrepfd{\signature}{\depset}$ can be solved in polynomial time if and only if $\depset=\emptyset$.
\end{lemma}

\begin{proof}
We will start by proving that if $\depset=\emptyset$ then $\maxsrepfd{\signature}{\depset}$ can be solved in polynomial time. Then, we will prove that if $\depset\neq\emptyset$, the problem $\maxsrepfd{\signature}{\depset}$ is NP-hard.
\paragraph*{The ``if'' direction}
Assume that $\depset=\emptyset$. In this case, $I$ is consistent and a C-repair of $I$ will just contain all of the facts in $I$. Thus, there is a polynomial time algorithm for solving $\maxsrepfd{\signature}{\depset}$.

\paragraph*{The ``only if'' direction}
Assume that $\depset\neq\emptyset$. We will prove that $\maxsrepfd{\signature}{\depset}$ is NP-hard. Note that in this case, $\depset$ cannot be a chain. Otherwise, $\depset$ contains a global minimum, which is either an FD of the form $\emptyset\rightarrow X$, in which case simplification $2$ can be applied to the schema, or an FD of the form $X\rightarrow Y$, where $X\neq \emptyset$, in which case it holds that $X\subseteq Z$ for each FD $Z\rightarrow W$ in $\depset$ and simplification $1$ can be applied to the schema. Thus, $\depset$ contains at least two local minima $X_1\rightarrow Y_1$ and $X_2\rightarrow Y_2$ (that is, no FD $Z\rightarrow W$ in $\depset$ is such that $Z\subset X_1$ or $Z\subset X_2$). One of the following holds:
\begin{enumerate}
\item $X_2^*\cap X_1 = \emptyset$. We divide this case into three subcases:
\begin{itemize}
\item $X_1^*\cap X_2^+= \emptyset$. In this case, Lemma~\ref{lemma:ac-bc-hard} and Lemma~\ref{lemma:disjoint-reduction} imply that $\maxsrepfd{\signature}{\depset}$ is NP-hard.
\item $X_1^* \cap X_2^*\neq \emptyset$ and $X_1^*\cap X_2 = \emptyset$. In this case, Lemma~\ref{lemma:ab-bc-hard} and Lemma~\ref{lemma:disjoint-left-reduction} imply that the problem $\maxsrepfd{\signature}{\depset}$ is NP-hard.
\item $X_1^*\cap X_2 \neq \emptyset$. In this case, Lemma~\ref{lemma:ab-bc-hard} and Lemma~\ref{lemma:disjoint-left-reduction} imply that $\maxsrepfd{\signature}{\depset}$ is NP-hard.
\end{itemize}
\item $X_2^*\cap X_1 \neq \emptyset$. We divide this case into three subcases:
\begin{itemize}
\item $X_1^*\cap X_2 \neq \emptyset$ and it holds that $(X_1\setminus X_2)\subseteq X_2^*$ and $(X_2\setminus X_1)\subseteq X_1^*$. In this case, $\depset$ contains at least one more local minimum. Otherwise, for every FD $Z\rightarrow W$ in $\depset$ it holds that either $X_1\subseteq Z$ or $X_2\subseteq Z$. If $X_1\cap X_2\neq\emptyset$, then we can apply simplification $1$ to the schema, using an attribute from $X_1\cap X_2$. If $X_1\cap X_2=\emptyset$, then we can apply simplification $3$ to the schema. In both cases, we get a contradiction to the fact that no simplifications can be applied to the schema. Thus, Lemma~\ref{lemma:abc-acb-bca-hard} and Lemma~\ref{lemma:triangle-reduction} imply that the problem $\maxsrepfd{\signature}{\depset}$ is NP-hard.
\item $X_1^*\cap X_2 \neq \emptyset$ and it holds that $(X_2\setminus X_1)\not\subseteq X_1^*$. In this case, Lemma~\ref{lemma:abc-cb-hard} and Lemma~\ref{lemma:last-reduction} imply that $\maxsrepfd{\signature}{\depset}$ is NP-hard.
\end{itemize}
\end{enumerate}
This concludes our proof of the lemma.
\end{proof}

Next, we prove the correctness of the algorithm $\algname{FindCRep}$.

\begin{lemma}\label{lemma:algorithm-proof}
Let $(\signature,\depset)$ be an FD schema, and let $I$ be an inconsistent instance of $\signature$. If $\maxsrepfd{\signature}{\depset}$ can be solved in polynomial time, then $\algname{FindCRep}(\signature,\depset)$ returns a C-repair of $I$. Otherwise, $\algname{FindCRep}(\signature,\depset)$ returns $\emptyset$.
\end{lemma}

\begin{proof}
We will prove the lemma by induction on $n$, the maximal number of simplifications that can be applied to $(\signature, \depset)$ (that is, there is a sequence of $n$ simplifications that can be applied to the schema until we reach to a point where we cannot apply any more simplifications, but there is no sequence of $m>n$ simplifications that can be applied to the schema until we reach such a point). The basis of the induction is $n=0$. In this case, Lemma~\ref{lemma:ind-basis} implies that $\maxsrepfd{\signature}{\depset}$ can be solved in polynomial time if and only if $\depset=\emptyset$ and indeed, if $\depset=\emptyset$ the condition of line~2 is satisfied and the algorithm returns a C-repair of $I$ (that contains all the facts), and if $\depset\neq\emptyset$, no condition is satisfied and the algorithm will return $\emptyset$ in line~10.

For the inductive step, we need to prove that if the claim is true for all $n=1,\dots,k-1$, it is also true for $n=k$. Let $(\signature\depset)$ be an FD schema, such that the maximal number of simplifications that can be applied to $(\signature, \depset)$ is $k$. The algorithm will apply some simplification $s$ to the schema. Clearly, the result is a schema $(\signature', \depset')$, such that the maximal number of simplifications that can be applied to $(\signature', \depset')$ is $k-1$. Thus, we know that if $\maxsrepfd{\signature'}{\depset'}$ cannot be solved in polynomial time, then $\algname{FindCRep}(\signature',\depset',J)$ will return $\emptyset$ for each $J$, and if $\maxsrepfd{\signature'}{\depset'}$ can be solved in polynomial time, then it can be solved using $\algname{FindCRep}$, and $\algname{FindCRep}(\signature',\depset',J)$ will return a C-repair of $J$ for each $J$. One of the following holds:

\begin{itemize}
\item $s$ is simplification $1$. In this case, the condition of line~4
  is satisfied and the subroutine $\algname{FindCRepS1}$ will be
  used. Lemma~\ref{lemma:s1-ptime} implies that there is a fact wise
  reduction from $(\pi_{\overline{\set{A_i}}}(\signature),
  \pi_{\overline{\set{A_i}}}(\depset))$ to $(\signature, \depset)$,
  thus if the problem $\maxsrepfd{\signature}{\depset}$ can be solved
  in polynomial time, the problem
  $\maxsrepfd{\pi_{\overline{\set{A_i}}}(\signature)}{\pi_{\overline{\set{A_i}}}(\depset)}$
  can be solved in polynomial time as well. Then, we know from the
  inductive assumption that the problem
  $\maxsrepfd{\pi_{\overline{\set{A_i}}}(\signature)}{\pi_{\overline{\set{A_i}}}(\depset)}$
  can be solved using $\algname{FindCRep}$ and
  Lemma~\ref{lemma:s1-ptime} implies that the problem
  $\maxsrepfd{\signature}{\depset}$ can be solved using
  $\algname{FindCRep}$ as well. In addition,
  Lemma~\ref{lemma:s1-ptime} implies that if the problem
  $\maxsrepfd{\pi_{\overline{\set{A_i}}}(\signature)}{\pi_{\overline{\set{A_i}}}(\depset)}$
  can be solved using $\algname{FindCRep}$, then so does
  $\maxsrepfd{\signature}{\depset}$. Thus, if
  $\maxsrepfd{\signature}{\depset}$ is a hard problem, then the
  problem
  $\maxsrepfd{\pi_{\overline{\set{A_i}}}(\signature)}{\pi_{\overline{\set{A_i}}}(\depset)}$
  is hard as well. We know from the inductive assumption that for each
  $J$, executing
  $\algname{FindCRep}(\pi_{\overline{\set{A_i}}}(\signature),
  \pi_{\overline{\set{A_i}}}(\depset),J)$ returns an empty set, and
  $\algname{FindCRepS1}(\signature,\depset,I)$ will return a union of
  empty sets, which is an empty set.

\item $s$ is simplification $2$. In this case, the condition of line~6
  is satisfied and the subroutine $\algname{FindCRepS2}$ will be
  used. Lemma~\ref{lemma:s2-ptime} implies that there is a fact wise
  reduction from $(\pi_{\overline{X}}(\signature),
  \pi_{\overline{X}}(\depset))$ to $(\signature, \depset)$, thus if
  the problem $\maxsrepfd{\signature}{\depset}$ can be solved in
  polynomial time, so does the problem
  $\maxsrepfd{\pi_{\overline{X}}(\signature)}{\pi_{\overline{X}}(\depset)}$. Then,
  we know from the inductive assumption that the problem
  $\maxsrepfd{\pi_{\overline{X}}(\signature)}{\pi_{\overline{X}}(\depset)}$
  can be solved using $\algname{FindCRep}$ and
  Lemma~\ref{lemma:s2-ptime} implies that the problem
  $\maxsrepfd{\signature}{\depset}$ can be solved using
  $\algname{FindCRep}$ as well. In addition,
  Lemma~\ref{lemma:s2-ptime} implies that if the problem
  $\maxsrepfd{\pi_{\overline{X}}(\signature)}{\pi_{\overline{X}}(\depset)}$
  can be solved using $\algname{FindCRep}$, then so does
  $\maxsrepfd{\signature}{\depset}$. Thus, if
  $\maxsrepfd{\signature}{\depset}$ is a hard problem, then the
  problem
  $\maxsrepfd{\pi_{\overline{X}}(\signature)}{\pi_{\overline{X}}(\depset)}$
  is hard as well. From the inductive assumption we conclude that for each $J$,
  $\algname{FindCRep}(\pi_{\overline{X}}(\signature),
  \pi_{\overline{X}}(\depset),J)$ returns an empty set, and
  $\algname{FindCRepS2}(\signature,\depset,I)$ will return an empty
  set as well (since it returns the instance that contains the most
  fact, and in this case all the instances are empty).

\item $s$ is simplification $3$. In this case, the condition of line~8
  is satisfied and the subroutine $\algname{FindCRepS3}$ will be
  used. Lemma~\ref{lemma:s3-ptime} implies that there is a fact wise
  reduction from the schema $(\pi_{\overline{X_1\cup
      X_2}}(\signature), \pi_{\overline{X_1\cup X_2}}(\depset))$ to
  the schema $(\signature, \depset)$, thus if the problem
  $\maxsrepfd{\signature}{\depset}$ can be solved in polynomial time,
  then the problem $\maxsrepfd{\pi_{\overline{X_1\cup
        X_2}}(\signature)}{\pi_{\overline{X_1\cup X_2}}(\depset)}$ can
  be solved in polynomial time as well. Then, we know from the
  inductive assumption that the algorithm $\algname{FindCRep}$ can be
  used to solve $\maxsrepfd{\pi_{\overline{X_1\cup
        X_2}}(\signature)}{\pi_{\overline{X_1\cup X_2}}(\depset)}$ and
  Lemma~\ref{lemma:s3-ptime} implies that
  $\maxsrepfd{\signature}{\depset}$ can be solved using
  $\algname{FindCRep}$ as well. In addition,
  Lemma~\ref{lemma:s3-ptime} implies that if
  $\maxsrepfd{\pi_{\overline{X_1\cup
        X_2}}(\signature)}{\pi_{\overline{X_1\cup X_2}}(\depset)}$ can
  be solved using $\algname{FindCRep}$, then
  $\maxsrepfd{\signature}{\depset}$ can also be solved using the
  algorithm. Thus, if the problem $\maxsrepfd{\signature}{\depset}$ is
  a hard problem, then the problem $\maxsrepfd{\pi_{\overline{X_1\cup
        X_2}}(\signature)}{\pi_{\overline{X_1\cup X_2}}(\depset)}$ is
  hard as well. From the inductive assumption we have that for each
  $J$, executing $\algname{FindCRep}(\pi_{\overline{X_1\cup
      X_2}}(\signature), \pi_{\overline{X_1\cup X_2}}(\depset),J)$
  returns $\emptyset$, and
  $\algname{FindCRepS3}(\signature,\depset,I)$ will return a union of
  empty sets, which is an empty set.
\end{itemize}
This concludes our proof of the correctness of algorithm $\algname{FindCRep}$.
\end{proof}

Note that the algorithm $\algname{FindCRep}$ starts with an FD schema $(\signature,\depset)$, and at each step it applies one simplification to the current schema. Finally, if no simplification can be applied to the schema and $\depset$ is empty, then the condition of line~2 will be satisfied, and the algorithm will eventually return a C-repair of $I$. If $\depset$ is not empty, the algorithm will return $\emptyset$. The proof of Theorem~\ref{thm:main-maxsrep} is straightforward based on this observation and on Lemma~\ref{lemma:algorithm-proof}.

\bibliographystyle{abbrv}
\bibliography{main} 
\end{document}